\documentclass[aps, pra, reprint, superscriptaddress, onecolumn, notitlepage, tightenlines, 11pt]{revtex4-2}

\pdfoutput=1

\usepackage{header}

\graphicspath{{graphics/}}

\begin{document}

\bibliographystyle{apsrev4-2}

\title{A Graph-Theoretic Framework for Free-Parafermion Solvability}

\author{Ryan L. Mann}
\email{mail@ryanmann.org}
\homepage{http://www.ryanmann.org}
\affiliation{Centre for Quantum Computation and Communication Technology}
\affiliation{Centre for Quantum Software and Information, School of Computer Science, Faculty of Engineering \& Information Technology, University of Technology Sydney, NSW 2007, Australia}

\author{Samuel J. Elman}
\email{samuel.elman@uts.edu.au}
\affiliation{Centre for Quantum Software and Information, School of Computer Science, Faculty of Engineering \& Information Technology,
University of Technology Sydney, NSW 2007, Australia}

\author{David R. Wood}
\email{david.wood@monash.edu}
\affiliation{School of Mathematics, Monash University, Melbourne, Australia}

\author{Adrian Chapman}
\thanks{Present address: Phasecraft Inc., Washington, D.C.}
\affiliation{Department of Materials, University of Oxford, Parks Road, Oxford OX1 3PH, United Kingdom}

\begin{abstract}
    We present a graph-theoretic characterisation of when a quantum spin model admits an exact solution via a mapping to free parafermions. Our characterisation is based on the concept of a frustration graph, which represents the commutation relations between Weyl operators of a Hamiltonian. We show that a quantum spin system has an exact free-parafermion solution if its frustration graph is an oriented indifference graph. Further, we show that if the frustration graph of a model can be dipath oriented via switching operations, then the model is integrable in the sense that there is a family of commuting independent set charges. Additionally, we establish an efficient algorithm for deciding whether this is possible. Our characterisation extends that given for free-fermion solvability. Finally, we apply our results to solve three qudit spin models.
\end{abstract}

\maketitle

{
\hypersetup{linkcolor=black}
\tableofcontents
}

\section{Introduction}
\label{section:Introduction}

Free fermions provide a fundamental model for understanding various phenomena in condensed-matter physics, quantum field theory, and high-energy physics. Their straightforward mathematical properties allow for precise calculations and predictions, for instance, in the electronic structure of materials. Additionally, certain spin systems can be mapped to free fermions, allowing for exact calculations and insights into these systems. The classical two-dimensional Ising model is an archetypal example of mapping spins to free fermions. Onsager's exact computation of the free energy~\cite{onsager1944crystal} can be understood as a consequence of expressing the transfer matrix in terms of bilinears of fermionic operators~\cite{kaufman1949crystal}, a process known as the Jordan--Wigner transformation~\cite{jordan1928uber}. This transformation allows for the exact computation of many properties, including the spectrum of the quantum Ising spin chain~\cite{schultz1964two}.

The mapping of Hamiltonian terms to fermionic monomials via a generator-to-generator approach has been extended beyond one-dimensional spin chains to higher dimensions~\cite{kitaev2006anyons, mandal2009exactly}. Further, it has been shown that a generator-to-generator mapping is possible if the frustration graph of the Hamiltonian is a line graph~\cite{chapman2020characterization, ogura2020geometric}. In this case, the root graph of the line graph describes the fermionic interactions of the model.

Remarkably, this graph-theoretic approach to finding free-fermion models extends beyond those solvable by a Jordan--Wigner transformation. A certain four-fermion Hamiltonian exhibits a free-fermion spectrum despite it not satisfying the conditions of Refs.~\cite{chapman2020characterization, ogura2020geometric}. Its spectrum nevertheless can be computed by explicit construction of raising and lowering operators~\cite{fendley2019free}. This construction utilises transfer operators formed from commuting symmetries of the model. This approach goes beyond the generator-to-generator method, with the mapping non-local and non-linear in terms of the original spins or fermions. A graph-theoretic method yields an entire class of Hamiltonians solvable by this construction. Namely, an exact free-fermion spectrum follows directly from a frustration graph when it obeys certain conditions~\cite{elman2021free, chapman2023unified}. This approach can be modified still further, yielding free-fermion spectra for several models not satisfying these conditions~\cite{fendley2024free}.

In view of these successes, it is natural to explore generalisations. One such generalisation involves considering qudit models instead of qubits. Such models are of physical interest because of their role in the study of topological order~\cite{alicea2016topological}. The Fradkin--Kadanoff transformation~\cite{fradkin1980disorder} provides a natural extension of the Jordan--Wigner transformation to qudit systems. The resulting parafermionic operators obey commutation relations involving a phase, which means that a qudit system expressed in terms of bilinear parafermions is typically not exactly solvable. The concept of free fermions does not naturally generalise, leading to the notion that free parafermions may not exist, even in one dimension. Even in integrable clock Hamiltonians, the correlators within the corresponding parafermionic conformal field theory do not obey Wick's Theorem~\cite{fendley2012parafermionic}. 

However, the integrable chiral Potts model shares many characteristics with free-fermion models~\cite{vongehlen1985zn, baxter1988free, auyang1997many, baxter2006challenge}. Further, the chiral clock chain exhibits a phase with robust topological order and an exact parafermionic edge zero mode~\cite{fendley2012parafermionic}. Baxter showed that a simple chiral non-Hermitian Hamiltonian has an exact spectrum that is a $\mathbb{Z}_d$ free-parafermion analogue of free-fermion spectra~\cite{baxter1989simple}. The derivation of this result is rather indirect and relies on results from the integrable two-dimensional classical chiral Potts model~\cite{vongehlen1985zn, auyang1997many, baxter2006challenge}. An associated two-dimensional classical model known as the $\tau_2$ model can be solved using the same techniques~\cite{bazhanov1990chiral, baxter2004transfer}. The transfer matrices of these two classical models are essentially inverses of each other.

This free-parafermion spectrum was obtained more directly in Ref.~\cite{fendley2014free}, along with the construction of a set of conserved charges. Specifically, shift operators that map between energy eigenstates were constructed and used to derive the spectrum. Although the commutation of the original parafermionic operators of Fradkin and Kadanoff yields only phases, the algebra of shift operators is much more intricate. The precise form of this generalised Clifford algebra was conjectured in Ref.~\cite{fendley2014free}, and conjectured to generalise to the $\tau_2$ classical models in Ref.~\cite{baxter2014tau2}. A proof of this was given in Refs.~\cite{au2014parafermions, au2016parafermions}. We establish a more direct proof following methods generalising those used to solve the aforementioned disguised free-fermion chains.

The direct solution of the free-parafermion models~\cite{fendley2014free} relies on the simple exchange algebra obeyed by the Hamiltonian terms~\cite{yamazaki1964projective, morris1967generalized}. This algebra was generalised to a broader family of $m$-body interaction chains, which satisfy a more general algebra~\cite{alcaraz2017energy, alcaraz2018anomalous, alcaraz2020integrable}. The algebra has since been used to explore a class of XY-type models with spectra composed of combinations of free parafermions~\cite{alcaraz2020free, alcaraz2021free}. Additionally, this approach has led to the development of an efficient numerical method for calculating the mass gaps associated with these free-particle modes~\cite{alcaraz2021powerful}. It has been shown that standard quantum Ising chains with inhomogeneous couplings can be designed to have the same spectra as a family of free-fermion quantum spin chains~\cite{alcaraz2023ising}. For a brief review of free parafermions, we refer the reader to Ref.~\cite{batchelor2023brief}.

The frustration-graph approach to free fermions~\cite{chapman2020characterization, elman2021free, chapman2023unified} does not automatically apply to free parafermions. There is currently no systematic way to determine when a qudit spin system can be exactly solved by mapping it to free parafermions. In this paper, we provide a graph-theoretic characterisation for when such a free-parafermion solution exists. Similar to the results on free fermions, our characterisation is based on the frustration graph of the model. However, due to the commutation relations of the qudit Weyl operators, our results consider oriented frustration graphs.

We show that if the frustration graph of a model can be dipath oriented via switching operations, then the model is integrable in the sense that there is a family of commuting independent set charges (Theorem~\ref{theorem:IndependentSetChargesCommute}). Our main result shows that if the frustration graph of the model is an oriented indifference graph, then it admits an exact free-parafermion solution (Theorem~\ref{theorem:FreeParafermionSolvability}). Additionally, we establish an efficient algorithm for deciding whether an oriented graph can be dipath oriented via switching operations (Theorem~\ref{theorem:DipathOrientationSwitchingAlgorithm}). Finally, we apply our results to solve three qudit spin models.

This paper is structured as follows. In Section~\ref{section:Preliminaries}, we introduce the necessary preliminaries. Then, in Section~\ref{section:FreeParafermionSolvabilityAndIntegrability}, we discuss our main results on free-parafermion solvability and integrability. In Section~\ref{section:DipathOrientationViaSwitchingOperations}, we present our algorithm for deciding dipath orientability via switching operations. We apply our results to solve three qudit spin models in Section~\ref{section:Applications}. Finally, we conclude in Section~\ref{section:ConclusionAndOutlook} with some remarks and open problems.

\section{Preliminaries}
\label{section:Preliminaries}

Let $\omega_d$ denote the primitive $d^\text{th}$ root of unity.

\subsection{Graph Theory}
\label{section:GraphTheory}

Let $G=(V,E)$ be a graph. We denote the open and closed neighbourhood of a vertex $v \in V$ by $\mathcal{N}(v)\coloneqq\{u\mid\{u,v\} \in E\}$ and $\mathcal{N}[v]\coloneqq\mathcal{N}(v)\cup\{v\}$, respectively. For a subset $U$ of $V$, the induced subgraph $G[U]$ is the subgraph of $G$ whose vertex set is $U$ and whose edge set consists of all edges in $G$ which have both endpoints in $U$. For convenience, we denote the induced subgraph $G[V{\setminus}U]$ by $G{\setminus}U$.

An \emph{independent set} of $G$ is a subset $U$ of $V$ with no edges between them. The \emph{independence number} $\alpha(G)$ of $G$ is the order of the largest independent set of $G$. Let $\mathcal{I}(G)$ denote the collection of all independent sets in $G$ and let $\mathcal{I}_i(G)$ denote the collection of all independent sets in $G$ of order $i\in\mathbb{N}$. A \emph{clique} of $G$ is a set of pairwise adjacent vertices in $G$. The \emph{claw} is the complete bipartite graph $K_{1,3}$. A graph is \emph{claw-free} if it does not include the claw as an induced subgraph.

A graph is \emph{chordal} if every induced cycle has order 3. A \emph{perfect elimination ordering} of a graph is an ordering of the vertices such that, for each vertex $v$, the neighbours of $v$ that precede $v$ in the ordering form a clique. A graph is chordal if and only if it has a perfect elimination ordering~\cite{rose1970triangulated}.

An \emph{indifference ordering} of a graph is an ordering of the vertices such that every induced path is either increasing or decreasing with respect to the ordering. A graph is an \emph{indifference graph} if it has an indifference ordering~\cite{looges1993optimal}. Note that since an indifference ordering is a perfect elimination ordering, indifference graphs are chordal. There is an efficient algorithm for deciding whether a graph is an indifference graph~\cite{looges1993optimal}.

Let $G$ be an oriented graph. A \emph{dipath} in $G$ is a path where all edges are oriented consistently. We say that $G$ is \emph{dipath oriented} if every induced path is a dipath. An \emph{oriented perfect elimination ordering} of $G$ is a perfect elimination ordering that is consistent with the orientation. An \emph{oriented indifference graph} is an indifference graph that is oriented consistently with an indifference ordering. Note that an oriented graph is an oriented indifference graph if and only if it is dipath oriented and has an oriented perfect elimination ordering.

The \emph{switching} of an oriented graph at a vertex $v$ is the operation of reversing all edges incident to $v$. The \emph{switching orbit} of an oriented graph is the set of all oriented graphs that can be obtained by applying a sequence of switching operations.

\begin{definition}[Independence polynomial]
    Let $G$ be a graph with vertex weights $\{w_v\}_{v \in V(G)}$. For $x\in\mathbb{C}$, the \emph{independence polynomial} $Z_G(x)$ is defined by
    \begin{equation}
        Z_G(x) \coloneqq \sum_{I\in\mathcal{I}(G)}\prod_{v \in I}xw_v. \notag
    \end{equation}
\end{definition}

There is an efficient algorithm for computing the independence polynomial when the graph is chordal~\cite{achlioptas2021local}.

\subsection{Qudit Models}
\label{section:QuditModels}

Fix $d\in\mathbb{Z}_{\geq2}$. Let $X_d\coloneqq\sum_{k\in\mathbb{Z}_d}\outerproduct{k+1}{k}$ and $Z_d\coloneqq\sum_{k\in\mathbb{Z}_d}\omega_d^k\outerproduct{k}{k}$ denote the \emph{generalised Pauli operators} in a $d$-dimensional space, often referred to as the shift and clock operators. The \emph{Weyl--Heisenberg group} in a $d$-dimensional space is generated by the operators
\begin{equation}
    \mathcal{W}_d \coloneqq \left\{\omega_d^{\frac{\alpha\beta}{2}} X_d^\alpha Z_d^\beta \mid \alpha,\beta\in\mathbb{Z}_d \right\}. \notag
\end{equation}
These operators form a complete basis for operators on a $d$-dimensional Hilbert space. A complete basis for operators on an $n$-qudit system can be formed by taking tensor products of these operators, that is,
\begin{equation}
    \mathcal{W}_d(n) \coloneqq \left\{\omega_d^{\frac{\alpha\cdot\beta}{2}}\bigotimes_{l\in\mathbb{Z}_n}X_d^{\alpha_l}Z_d^{\beta_l}\mid\alpha,\beta\in\mathbb{Z}_d^n \right\}. \notag
\end{equation}
The elements of $\mathcal{W}_d(n)$ are called \emph{Weyl operators}. We consider $n$-qudit systems with Hamiltonians written in the generalised Pauli basis
\begin{equation}
    H = \sum_{v \in V}h_v, \notag
\end{equation}
where $V$ is subset of $\mathbb{Z}_d^{2n}$ and $h_v=b_vw_v$ for $b_v\in\mathbb{R}{\setminus}\{0\}$ and $w_v$ an element of $\mathcal{W}_d(n)$ determined by the label $v$ in the natural way. Note that the Hamiltonian is not necessarily Hermitian. However, it can be made Hermitian by adding its Hermitian conjugate. We shall restrict our attention to Hamiltonians that satisfy $h_uh_v=\omega_d^kh_vh_u$ with $k$ in $\{-1,0,1\}$ for all $u$ and $v$ in $V$. This restriction is necessary for our definition of the frustration graph and is always true when $d=2$ or $d=3$.

\begin{definition}[Frustration graph]
    The \emph{frustration graph} of a Hamiltonian is the oriented graph $G=(V,E)$ with vertex set $V$ and edge set $E=\left\{(u,v) \mid h_uh_v=\omega_dh_vh_u\right\}$, where the notation $(u,v)$ denotes a directed edge from vertex $u$ to vertex $v$.
\end{definition}

\subsection{Parafermions}
\label{subsection:parafermions}

Parafermions are the $\mathbb{Z}_d$-symmetric generalisation of fermions. They obey a more complex exchange algebra than fermions. For instance, while Majorana fermions $\{c_j\}_j$ satisfy the modified canonical anticommutation relations:
\begin{align}
    c_j^2 &= 1, & c_j &= c_j^\dagger, & c_jc_k &= (-1)^{\delta_{j,k}+1}c_kc_j, \notag
\end{align}
the Majorana-like parafermions $\{\gamma_j\}_j$ satisfy:
\begin{align}
    \gamma_j^d &= 1, & \gamma_j^{d-1} &= \gamma_j^\dagger, & \gamma_j^m\gamma_k^n &= \omega_d^{mn\operatorname{sgn}(k-j)}\gamma_k^n\gamma_j^m, \notag
\end{align}
for $m,n\in\mathbb{Z}_d$. Similarly, the concept of the free parafermion is a $\mathbb{Z}_d$ generalisation of the free fermion. While free fermions have a real energy spectrum of the form:
\begin{equation}
    \left\{\sum_{k \in [n]}(-1)^{x_k}\varepsilon_k \mid x\in\mathbb{Z}_2^n\right\}, \notag
\end{equation}
free parafermions have a complex energetic spectrum of the form:
\begin{equation}
    \left\{\sum_{k \in [n]}\omega_d^{x_k}\varepsilon_k \mid x\in\mathbb{Z}_d^n\right\}, \notag
\end{equation}
where $n$ denotes the number of particles and $\{\varepsilon_k\}_k$ denotes the single-particle energies.

We note that for free parafermions, the single-particle energies may themselves be complex, making the identification of free parafermion models from the spectrum more challenging than for fermions. For free fermions, the ability to map the Pauli strings of a spin-$1/2$ model to fermionic bilinear monomials is a sufficient condition to identify a qubit model as a free-fermion model, for instance, via a Jordan--Wigner transformation. However, this criterion does not apply to free parafermions~\cite{stoudenmire2015assembling}. In this work, we provide a graph-theoretic criterion for identifying free parafermions in qudit models, generalising the work of Ref.~\cite{chapman2023unified} to qudit models.

\section{Free-Parafermion Solvability and Integrability}
\label{section:FreeParafermionSolvabilityAndIntegrability}

In this section, we establish our results on free-parafermion solvability and integrability. We begin by introducing a family of commuting charges defined by sets of Hamiltonian terms that commute, or equivalently the independent sets of the frustration graph.

\begin{definition}[Independent set charges]
    Let $G$ be a graph and let $H$ be a Hamiltonian. For $0 \leq i \leq \alpha(G)$, the \emph{independent set charge} $Q_G^{(i)}$ is defined by
    \begin{equation}
        Q_G^{(i)} \coloneqq \sum_{S\in\mathcal{I}_i(G)}h_S, \notag
    \end{equation}
    where $h_S\coloneqq\prod_{v \in S}h_v$. Note that $Q_G^{(0)}=I$ and $Q_G^{(1)}=H$.
\end{definition}

We now introduce the transfer operator, which can be viewed as an operator-valued counterpart to the independence polynomial of a graph.
\begin{definition}[Transfer operator]
    Let $G$ be a graph and let $H$ be a Hamiltonian. For $x\in\mathbb{C}$, the \emph{transfer operator} $T_G(x)$ is defined by
    \begin{equation}
        T_G(x) \coloneqq \sum_{k=0}^{\alpha(G)}(-x)^kQ_G^{(k)}. \notag
    \end{equation}
\end{definition}

The following two propositions illustrate how the transfer operator and the independence polynomial obey similar recursion relations.
\begin{proposition}
    \label{proposition:VertexRecursionRelation}
    Let $G$ be a graph and let $v$ be a vertex in $G$. The transfer operator $T_G(x)$ and independence polynomial $Z_G(x)$ satisfy the following recursion relations.
    \begin{equation} 
        T_G(x) = T_{G\setminus\{v\}}(x)-xh_vT_{G\setminus\mathcal{N}[v]}(x), \notag
    \end{equation}
    \begin{equation} 
        Z_G(x) = Z_{G\setminus\{v\}}(x)+xw_vZ_{G\setminus\mathcal{N}[v]}(x). \notag
    \end{equation}
\end{proposition}

\begin{proposition}
    \label{proposition:CliqueRecursionRelation}
    Let $G$ be a graph and let $K$ be a clique in $G$. The transfer operator $T_G(x)$ and independence polynomial $Z_G(x)$ satisfy the following recursion relations.
    \begin{equation} 
        T_G(x) = T_{G \setminus K}(x)-x\sum_{v \in K}h_vT_{G\setminus\mathcal{N}[v]}(x), \notag
    \end{equation}
    \begin{equation} 
        Z_G(x) = Z_{G \setminus K}(x)+x\sum_{v \in K}w_vZ_{G\setminus\mathcal{N}[v]}(x). \notag
    \end{equation}
\end{proposition}

We now show when the independent set charges are commuting symmetries of the Hamiltonian. In Refs.~\cite{elman2021free, chapman2023unified}, it was shown that for qubit models, the independent set charges commute if the frustration graph is claw-free. However, for qudit models, the condition is more complicated due to the complex commutation relations of the Weyl operators.

\begin{theorem}
    \label{theorem:IndependentSetChargesCommute}
    Fix $d\in\mathbb{Z}_{\geq2}$. Let $H$ be a qudit Hamiltonian with dipath-oriented frustration graph $G$. The independent set charges $\{Q_G^{(i)}\}$ pairwise commute.
\end{theorem}

\begin{proof}
    By definition,
    \begin{equation}
        [Q_G^{(i)},Q_G^{(j)}] = \sum_{R\in\mathcal{I}_i(G)}\sum_{S\in\mathcal{I}_j(G)}[h_R,h_S]. \notag
    \end{equation}
    Let $R$ and $S$ be two independent sets of $G$. Note that since $G$ is dipath oriented, it is claw-free. Therefore, the induced subgraph $G[R \cup S$] is a disjoint union of paths and even cycles. Let $\mathcal{C}$ denote the set of all maximally connected components of $G[R \cup S$]. Observe that, for $C\in\mathcal{C}$, $[h_{R \cap C},h_{S \cap C}]=0$ unless $C$ is an odd path. Let $\uparrow_R$ and $\uparrow_S$ denote the respective number of odd paths in $\mathcal{C}$ starting with an element of $R$ and $S$ in the dipath orientation of $G$. Then,
    \begin{equation}
        [h_R,h_S] = \left(1-\omega_d^{\uparrow_S-\uparrow_R}\right)h_Rh_S. \notag
    \end{equation}
    Let $P$ denote the set of all vertices that are contained in an odd path of $\mathcal{C}$. Now consider the independent sets $\bar{R}$ and $\bar{S}$ formed from $R$ and $S$ by exchanging vertices in $P$, that is, $\bar{R} \coloneqq R \Delta P$ and $\bar{S} \coloneqq S \Delta P$. Then,
    \begin{align}
        [h_{\bar{R}},h_{\bar{S}}] &= \left(1-\omega_d^{\uparrow_{\bar{S}}-\uparrow_{\bar{R}}}\right)h_{\bar{R}}h_{\bar{S}} \notag \\
        &= \left(1-\omega_d^{\uparrow_R-\uparrow_S}\right)h_{R{\setminus}P}h_{P{\setminus}R}h_{P{\setminus}S}h_{S{\setminus}P} \notag \\
        &= \left(1-\omega_d^{\uparrow_R-\uparrow_S}\right)\omega_d^{\uparrow_S-\uparrow_R}h_{R{\setminus}P}h_{P{\setminus}S}h_{P{\setminus}R}h_{S{\setminus}P} \notag \\
        &= -\left(1-\omega_d^{\uparrow_S-\uparrow_R}\right)h_Rh_S. \notag
    \end{align}
    Hence,
    \begin{equation}
        [h_R,h_S]+[h_{\bar{R}},h_{\bar{S}}] = 0. \notag
    \end{equation}
    It then follows that the independent set charges $\{Q_G^{(i)}\}$ pairwise commute, completing the proof.
\end{proof}

Theorem~\ref{theorem:IndependentSetChargesCommute} implies that a qudit Hamiltonian with a frustration graph that can be dipath oriented is integrable. Additionally, we can conclude that the independent set charges will commute in this manner only for one-dimensional models~\cite{brandstadt2003linear}. In Section~\ref{section:DipathOrientationViaSwitchingOperations}, we establish an efficient algorithm for deciding whether a frustration graph can be dipath oriented under switching operations. In this context, the switching operation corresponds to replacing a Hamiltonian term with its complex conjugate. We have the immediate Corollary of Theorem~\ref{theorem:IndependentSetChargesCommute}.

\begin{corollary}
    \label{corollary:HamiltonianTransferOperatorCommutator}
    Fix $d\in\mathbb{Z}_{\geq2}$. Let $H$ be a qudit Hamiltonian with dipath-oriented frustration graph $G$. The transfer operator satisfies $[H,T_G(x)]=0$ and $[T_G(x),T_G(y)]=0$ for all $x,y\in\mathbb{C}$.
\end{corollary}

Baxter's model~\cite{baxter1988free, baxter1989simple} was solved using a mapping to free parafermions in Ref.~\cite{fendley2014free}. A similar approach was used to identify the free-fermion spectrum of a model that is not solvable via a Jordan--Wigner transformation~\cite{fendley2019free}. This was generalised in Refs.~\cite{alcaraz2020free, alcaraz2020integrable} to a specific family of qudit models, and named equipartition indifference models in Ref.~\cite{elman2021free}. The approach of Ref.~\cite{fendley2019free} can also be generalised to Baxter's model and the more general $\tau_2$ model~\cite{fendley2023personal}. This is similar to the approach used for free-fermion solutions to many-body
models defined on spin-$1/2$ particles~\cite{chapman2020characterization, elman2021free, chapman2023unified}. We now characterise qudit models that can be solved via a mapping to free parafermions based on a graph-theoretic criterion. Our main result may be stated as follows.

\begin{theorem}
    \label{theorem:FreeParafermionSolvability}
    Fix $d\in\mathbb{Z}_{\geq2}$. Let $H$ be a qudit Hamiltonian with an oriented indifference frustration graph. There exist single-particle energies $\{\varepsilon_k\}$ and parafermionic modes $\{\psi_{r,k}\}$ such that
    \begin{equation}
        H = \sum_{k=1}^{\alpha(G)}\sum_{r\in\mathbb{Z}_d}\omega_d^r\varepsilon_k\prod_{p=1}^d\psi_{r-p,k}. \notag
    \end{equation}
    Further, the single-particle energies $\{\varepsilon_k\}$ satisfy $Z_G\left(-\varepsilon_k^{-d}\right)=0$ for all $k\in[\alpha(G)]$.
\end{theorem}

As in the free-fermion case, a general solution to a qudit Hamiltonian requires that the product of transfer operators is equal to the independence polynomial of the frustration graph~\cite{elman2021free}. While the independent set charges commute in models where the frustration graphs are dipath oriented, additional conditions are necessary for the independence polynomial to be expressed as a product of transfer operators. Note that the dipath orientation condition allows for periodic boundary conditions, while the oriented indifference condition does not.

\begin{lemma}[{restate=[name=restatement]IndependencePolynomialTransferMatrices}]
    \label{lemma:IndependencePolynomialTransferMatrices}
    Fix $d\in\mathbb{Z}_{\geq2}$. Let $H$ be a qudit Hamiltonian with oriented indifference frustration graph $G$. Then,
    \begin{equation}
        Z_G\left(-u^d\right) = \prod_{m\in\mathbb{Z}_d}T_G\left(u\omega_d^{-m}\right). \notag
    \end{equation}
\end{lemma}

We prove Lemma~\ref{lemma:IndependencePolynomialTransferMatrices} in Appendix~\ref{section:IndependencePolynomialTransferMatrices}. To solve the free parafermion model, we need to determine the single-particle energies. As in the free-fermion case, we can identify the single-particle energies of the model using the roots of the independence polynomial. To achieve this, we introduce the concept of simplicial and parafermionic modes.

\begin{definition}[Single-particle energies]
    Fix $d\in\mathbb{Z}_{\geq2}$. Let $H$ be a qudit Hamiltonian with oriented indifference frustration graph $G$. For $k\in[\alpha(G)]$, the \emph{single-particle energies} $\{\varepsilon_k\}$ are defined by
    \begin{equation}
        Z_G\left(-\varepsilon_k^{-d}\right) = 0. \notag
    \end{equation}
\end{definition}

\begin{definition}[Simplicial mode]
    Fix $d\in\mathbb{Z}_{\geq2}$. Let $H$ be a qudit Hamiltonian with oriented indifference frustration graph $G$. A \emph{simplicial mode} with respect to an oriented perfect elimination ordering of $G$ is a Weyl operator $\chi$ such that, for $u \in V(G)$,
    \begin{equation}
        [\chi,h_u] = \delta_{u,v}(1-\omega_d)\chi h_u, \notag
    \end{equation}
    where $v$ denotes the last vertex in the oriented perfect elimination ordering of $G$.
\end{definition}

\begin{definition}[Parafermionic mode]
    Fix $d\in\mathbb{Z}_{\geq2}$. Let $H$ be a qudit Hamiltonian with oriented indifference frustration graph $G$. Further let $\chi$ be a simplicial mode with respect to an oriented perfect elimination ordering of $G$ and let $v$ denote the last vertex in the oriented perfect elimination ordering of $G$. 
    For $p\in\mathbb{Z}_d$ and $k\in[\alpha(G)]$, the \emph{parafermionic mode} $\psi_{p,k}$ of degree $p$ associated with mode $k$ is defined by
    \begin{equation}
        \psi_{p,k} \coloneqq \frac{1}{N_k}T_G^p\left(\varepsilon_k^{-1}\right) \chi T_G^{\succ p}\left(\varepsilon_k^{-1}\right), \notag
    \end{equation}
    where
    \begin{equation}
        N_k \coloneq (1-\omega_d)\left[dZ_{G\setminus\{v\}}\left(-\varepsilon_k^{-d}\right)\left(\varepsilon_k\frac{\partial Z_G\left(-\varepsilon^{-d}\right)}{\partial\varepsilon}\bigg|_{\varepsilon=\varepsilon_k}\right)^{d-1}\right]^{\frac{1}{d}}, \notag
    \end{equation}
    and
    \begin{equation}
        T_G^p(u) \coloneqq T_G\left(u\omega_d^{-p}\right) \quad\text{and}\quad T_G^{\succ p}(u) \coloneqq \prod_{m=1}^{d-1}T_G\left(u\omega_d^{-p-m}\right). \notag
    \end{equation}
\end{definition}

The following lemma establishes the commutation relation between a qudit Hamiltonian with an oriented indifference frustration graph and its corresponding parafermionic modes.
\begin{lemma}[{restate=[name=restatement]HamiltonianParafermionicModeCommutator}]
    \label{lemma:HamiltonianParafermionicModeCommutator}
    Fix $d\in\mathbb{Z}_{\geq2}$. Let $H$ be a qudit Hamiltonian with oriented indifference frustration graph $G$. The single-particle energies $\{\varepsilon_k\}$ and parafermionic modes $\{\psi_{p,k}\}$ satisfy
    \begin{equation}
        [H,\psi_{p,k}] = (1-\omega_d)\omega_d^{p+1}\varepsilon_k\psi_{p,k}. \notag
    \end{equation}
\end{lemma}

We prove Lemma~\ref{lemma:HamiltonianParafermionicModeCommutator} in Appendix~\ref{section:HamiltonianParafermionicModeCommutator}.

\begin{lemma}[{restate=[name=restatement]ParafermionicModeCommutator}]
    \label{lemma:ParafermionicModeCommutator}
    Fix $d\in\mathbb{Z}_{\geq2}$. Let $H$ be a qudit Hamiltonian with oriented indifference frustration graph $G$. The single-particle energies $\{\varepsilon_k\}$ and parafermionic modes $\{\psi_{p,k}\}$ satisfy
    \begin{equation}
        \left(\omega_d^p\varepsilon_k-\omega_d^{q+1}\varepsilon_l\right)\psi_{p,k}\psi_{q,l} = -\left(\omega_d^q\varepsilon_l-\omega_d^{p+1}\varepsilon_k\right)\psi_{q,l}\psi_{p,k}, \notag
    \end{equation}
    for $k$ and $l$ distinct, and 
    \begin{equation}
        \psi_{p,k}\psi_{q,k} = 0, \notag
    \end{equation}
    for $p$ and $q+1$ distinct.
\end{lemma}

We prove Lemma~\ref{lemma:ParafermionicModeCommutator} in Appendix~\ref{section:ParafermionicModeCommutator}. In the remainder of this section, we focus on the mathematical details needed to determine the normalisation of the parafermionic modes and to complete the proof of Theorem~\ref{theorem:FreeParafermionSolvability}. To establish the normalisation, we first define the projection operators as products of the parafermionic modes. These operators can be understood as generalisations of the raising and lowering operators commonly used in qubit systems.

\begin{definition}[Projection operators]
    Fix $d\in\mathbb{Z}_{\geq2}$. Let $H$ be a qudit Hamiltonian with oriented indifference frustration graph $G$. For $r\in\mathbb{Z}_d$ and $k\in[\alpha(G)]$, the \emph{projection operator} $\mathcal{P}_{r,k}$ is defined by
    \begin{equation}
        \mathcal{P}_{r,k} \coloneqq \prod_{p=1}^d\psi_{r-p,k}. \notag
    \end{equation}
\end{definition}

\begin{lemma}[{restate=[name=restatement]ProjectorOperatorRelations}]
    \label{lemma:ProjectorOperatorRelations}
    Fix $d\in\mathbb{Z}_{\geq2}$. Let $H$ be a qudit Hamiltonian with oriented indifference frustration graph $G$. The single-particle energies $\{\varepsilon_k\}$ and projection operators $\{\mathcal{P}_{r,k}\}$ satisfy
    \begin{equation}
        \mathcal{P}_{r,k} = \left[\frac{\partial Z_G\left(-\varepsilon^{-d}\right)}{\partial\varepsilon}\bigg|_{\varepsilon=\varepsilon_k}\right]^{-1}T_G^{\succ r}\left(\varepsilon_k^{-1}\right)\frac{\partial T_G^r\left(\varepsilon^{-1}\right)}{\partial\varepsilon}\bigg|_{\varepsilon=\varepsilon_k}, \notag
    \end{equation}
    $\mathcal{P}_{r,k}^2=\mathcal{P}_{r,k}$, and $\mathcal{P}_{r,k}\mathcal{P}_{s,k}=0$ for $r$ and $s$ distinct.
\end{lemma}

We prove Lemma~\ref{lemma:ProjectorOperatorRelations} in Appendix~\ref{section:ProjectorOperatorRelations}.

\begin{lemma}[{restate=[name=restatement]HamiltonianProjectorOperatorRelation}]
    \label{lemma:HamiltonianProjectorOperatorRelation}
    Fix $d\in\mathbb{Z}_{\geq2}$. Let $H$ be a qudit Hamiltonian with oriented indifference frustration graph $G$. The single-particle energies $\{\varepsilon_k\}$ and projection operators $\{\mathcal{P}_{r,k}\}$ satisfy
    \begin{equation}
        H = \sum_{k=1}^{\alpha(G)}\sum_{r\in\mathbb{Z}_d}\omega_d^r\varepsilon_k\mathcal{P}_{r,k}. \notag
    \end{equation}
\end{lemma}

We prove Lemma~\ref{lemma:HamiltonianProjectorOperatorRelation} in Appendix~\ref{section:HamiltonianProjectorOperatorRelation}. Combining Lemma~\ref{lemma:HamiltonianParafermionicModeCommutator}, Lemma~\ref{lemma:ProjectorOperatorRelations}, and Lemma~\ref{lemma:HamiltonianProjectorOperatorRelation} proves Theorem~\ref{theorem:FreeParafermionSolvability}.

\section{Dipath Orientation via Switching Operations}
\label{section:DipathOrientationViaSwitchingOperations}

In this section, we establish an efficient algorithm for deciding whether an oriented graph can be dipath oriented via switching operations.

\begin{theorem}
    \label{theorem:DipathOrientationSwitchingAlgorithm}
    Let $G$ be an oriented graph. There is an polynomial-time algorithm to decide whether there exists a dipath-oriented graph in the switching orbit of $G$.
\end{theorem}

\begin{proof}
    Let $\mathcal{P}_2$ denote the set of all induced paths of length $2$ in $G$. Further let $x_e$ be an indicator variable for whether $e$ is in $E(G)$ and let $s_v$ be an indicator variable for whether a switching has been applied to vertex $v$ in $V(G)$. Consider the following system of equations over $\mathbb{F}_2$:
    \begin{equation}
        \left\{x_{(u,v)}+x_{(v,w)}+s_u+s_w\right\}_{(u,v,w)\in\mathcal{P}_2}. \notag
    \end{equation}
    Note that each equation evaluates to zero if and only if the corresponding path is a dipath. To ensure the graph is dipath oriented, it is necessary and sufficient that every induced path of length 2 in the graph is dipath oriented. This condition is satisfied if and only if the system of equations for all such paths evaluates to zero. Thus, the system of equations evaluates to zero if and only if the graph is dipath oriented. Since this is a system of linear equations over $\mathbb{F}_2$, it can be solved in polynomial time using Gaussian elimination, completing the proof.
\end{proof}

\section{Applications}
\label{section:Applications}

In this section, we apply our framework to three models. Firstly, we apply our framework to Baxter's model~\cite{baxter1989simple}. The model is defined on a line with respect to a parameter $n\in\mathbb{N}$. The Hamiltonian is given by
\begin{equation}
    H_n = \sum_{j=1}^naZ^\dagger_{j}Z_{j+1}+\sum_{j=1}^{n+1}bX_{j}, \notag
\end{equation}
where $a$ and $b$ are real parameters. The frustration graph of the model is a dipath, and therefore, an oriented indifference graph. Hence, Theorem~\ref{theorem:IndependentSetChargesCommute} and Theorem~\ref{theorem:FreeParafermionSolvability} apply.

We now apply our framework to the class of one-dimensional qudit spin models introduced in Ref.~\cite{alcaraz2020integrable}, which are solvable via a mapping to free parafermions. This class of models generalise the free-fermion models of Ref.~\cite{fendley2019free} to qudit systems. The class is defined with respect to parameters $n\in\mathbb{N}$ and $p\in\mathbb{N}$. The Hamiltonian is given by
\begin{equation}
    H_{n,p} = \sum_{j=1}^na_jX_jX_{j+1} \ldots X_{j+p-1}Z_{j+p}, \notag
\end{equation}
where $\{a_j\}$ are real parameters.

\begin{figure}[ht!]
    \centering
    \includegraphics[width=0.9\textwidth]{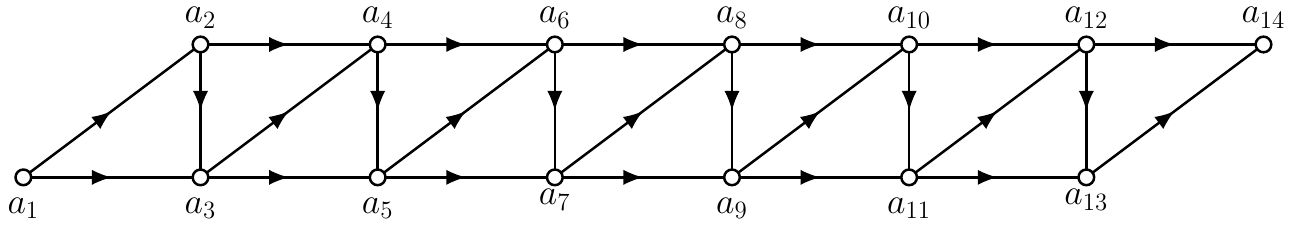}
    \caption{The frustration graph for the one-dimensional model introduced in Ref.~\cite{alcaraz2020integrable} with $n=14$ and $p=2$.}
    \label{figure:alcarazpimentamodel}
\end{figure}

The frustration graph of the model with $n=14$ and $p=2$ is shown in Fig.~\ref{figure:alcarazpimentamodel}. The frustration graph is an oriented indifference graph for all $n\in\mathbb{N}$ and all $p\in\mathbb{N}$. Hence, Theorem~\ref{theorem:IndependentSetChargesCommute} and Theorem~\ref{theorem:FreeParafermionSolvability} apply.

Finally, we introduce a one-dimensional qudit spin model and apply our framework. The model is defined on a line with respect to a parameter $n\in\mathbb{N}$, with qudits with $d=3$ located at sites $j$, $j+\frac{1}{3}$, and $j+\frac{2}{3}$ for $j\in[n]$. The Hamiltonian is given by
\begin{align}
    H_n = \sum_{j=1}^n&\left[aX_{j}Z^\dagger_{j+\frac{1}{3}}+b\omega_3(XZ^\dagger)_{j+\frac{1}{3}}Z_{j+\frac{2}{3}}+cX_{j+\frac{1}{3}}Z_{j+\frac{2}{3}}\right. \notag \\
    &\quad+ \left.dZ^\dagger_{j+\frac{1}{3}}X_{j+\frac{2}{3}}+e\omega_3^\frac{1}{2}Z^\dagger_{j+\frac{1}{3}}(XZ)_{j+\frac{2}{3}}+fZ_{j+\frac{2}{3}}Z^\dagger_{j+1}\right], \notag
\end{align}
where $a$, $b$, $c$, $d$, $e$, and $f$ are real parameters.

\begin{figure}[ht!]
    \centering
    \includegraphics[width=0.9\textwidth]{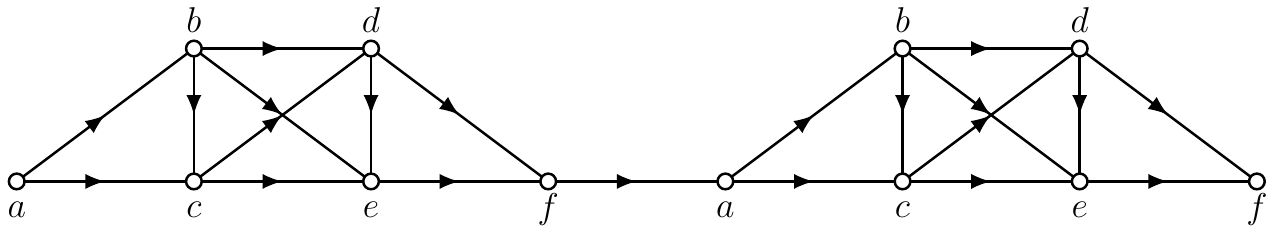}
    \caption{The frustration graph for the one-dimensional model with $n=2$.}
    \label{figure:exampleapplication}
\end{figure}

The frustration graph of the model with $n=2$ is shown in Fig.~\ref{figure:exampleapplication}. The frustration graph is an oriented indifference graph for all $n\in\mathbb{N}$. Hence, Theorem~\ref{theorem:IndependentSetChargesCommute} and Theorem~\ref{theorem:FreeParafermionSolvability} apply. By Theorem~\ref{theorem:IndependentSetChargesCommute}, the model is integrable in the sense that there is a family of commuting independent set charges. By Theorem~\ref{theorem:FreeParafermionSolvability}, the model admits an explicit free-parafermion solution, and the single-particle energies are identified with the roots of the independence polynomial. The independence polynomial can be computed efficiently via the algorithm of Ref.~\cite{achlioptas2021local} or by solving a recurrence relation obtained by applying Proposition~\ref{proposition:CliqueRecursionRelation}.

\section{Conclusion \& Outlook}
\label{section:ConclusionAndOutlook}

We have developed a graph-theoretic framework for free-parafermion solvability and integrability. Our results show that if the frustration graph of a model is an oriented indifference graph, then it admits an exact free-parafermion solution. Further, we have shown that if the frustration graph of a model can be dipath oriented via switching operations, then the model is integrable in the sense that there is a family of commuting independent set charges. Additionally, we have established an efficient algorithm for determining whether this is possible. Finally, we applied our results to solve three qudit spin models.

There is a well-understood connection between free fermions and matchgates~\cite{knill2001fermionic, terhal2002classical}. Research has investigated disguised free-fermion circuits~\cite{pozsgay2024quantum, vona2025exact}. Exploring generalisations of matchgates and extending these results to free parafermions would be interesting. Further, while there exists an efficient algorithm for deciding whether a graph is an indifference graph~\cite{looges1993optimal}, it would be interesting to develop an efficient algorithm for deciding whether an oriented graph can be transformed into an oriented indifference graph via switching operations.

\section*{Acknowledgements}

We thank Paul Fendley, Catherine Greenhill, and Jannis Ruh for helpful discussions. RLM was supported by the ARC Centre of Excellence for Quantum Computation and Communication Technology (CQC2T), project number CE170100012. SJE was supported with funding from the Defense Advanced Research Projects Agency under the Quantum Benchmarking (QB) program under award no. HR00112230007, HR001121S0026, and HR001122C0074 contracts. DRW is supported by the Australian Research Council. AC acknowledges support from EPSRC under agreement EP/T001062/1, and from EU H2020-FETFLAG-03-2018 under grant agreement no. 820495 (AQTION). The views, opinions and/or findings expressed are those of the authors and should not be interpreted as representing the official views or policies of the Department of Defense or the U.S. Government.

\appendix

\section{Proof of Lemma~\ref*{lemma:IndependencePolynomialTransferMatrices}}
\label{section:IndependencePolynomialTransferMatrices}

\IndependencePolynomialTransferMatrices*

To prove Lemma~\ref{lemma:IndependencePolynomialTransferMatrices}, we first require the following lemma.

\begin{lemma}
    \label{lemma:IndependencePolynomialBaseCaseAlgebraicIdentity}
    Fix $d\in\mathbb{Z}^+$. Then, for $l\in\mathbb{Z}_{d+1}$,
    \begin{equation}
        (-1)^l\sum_{\substack{S\subseteq\mathbb{Z}_d \\ \abs{S}=l}}\prod_{s \in S}\omega_d^{-s} = 
        \begin{cases}
            +1, & \text{if $l=0$}, \\
            -1, & \text{if $l=d$}, \\
            0, & \text{otherwise}.
      \end{cases} \notag
    \end{equation}
\end{lemma}

\begin{proof}
    The proof follows by observing that the polynomials
    \begin{equation}
        \prod_{m\in\mathbb{Z}_d}\left(1-u\omega_d^{-m}\right) = \sum_{l=0}^d(-u)^l\sum_{\substack{S\subseteq\mathbb{Z}_d \\ \abs{S}=l}}\prod_{s \in S}\omega_d^{-s} \quad\text{and}\quad 1-u^d \notag  
    \end{equation}
    are equal since they have identical roots and identical constant term.
\end{proof}

\begin{proof}[Proof of Lemma~\ref*{lemma:IndependencePolynomialTransferMatrices}.]
    We prove the statement by induction on an oriented perfect elimination ordering of $G$. First, we prove the base case of a graph $G_0$ comprising a single vertex $v_0$. By applying Lemma~\ref{lemma:IndependencePolynomialBaseCaseAlgebraicIdentity},
    \begin{equation}
        \prod_{m\in\mathbb{Z}_d}T_{G_0}\left(u\omega_d^{-m}\right) = \prod_{m\in\mathbb{Z}_d}\left(1-u\omega_d^{-m}h_{v_0}\right) = \sum_{l=0}^d\left(-uh_{v_0}\right)^l\sum_{\substack{S\subseteq\mathbb{Z}_d \\ \abs{S}=l}}\prod_{s \in S}\omega_d^{-s} = 1-(ub_v)^d = Z_{G_0}\left(-u^d\right). \notag
    \end{equation}
    Let $v$ denote the last vertex in the oriented perfect elimination ordering of $G$. We now prove the statement is true for $G$ under the assumption that it is true for $G{\setminus}\{v\}$. Note that since $v$ is the last vertex in an oriented perfect elimination ordering of $G$, $N(v)$ is a clique. By applying the inductive hypothesis and Proposition~\ref{proposition:CliqueRecursionRelation}, we obtain
    \begin{align}
        Z_{G\setminus\{v\}}\left(-u^d\right) &= \prod_{m\in\mathbb{Z}_d}T_{G\setminus\{v\}}\left(u\omega_d^{-m}\right) \notag \\
        &= \prod_{m\in\mathbb{Z}_d}\left[T_{G\setminus\mathcal{N}[v]}\left(u\omega_d^{-m}\right)-u\omega_d^{-m}\sum_{w\in\mathcal{N}(v)}h_wT_{G\setminus\mathcal{N}[w]}\left(u\omega_d^{-m}\right)\right]. \notag
    \end{align}
    Now, by applying Proposition~\ref{proposition:VertexRecursionRelation} and Proposition~\ref{proposition:CliqueRecursionRelation}, we have
    \begin{align}
        \prod_{m\in\mathbb{Z}_d}T_G\left(u\omega_d^{-m}\right) &= \prod_{m\in\mathbb{Z}_d}\left[T_{G\setminus\{v\}}\left(u\omega_d^{-m}\right)-u\omega_d^{-m}h_vT_{G\setminus\mathcal{N}[v]}\left(u\omega_d^{-m}\right)\right] \notag \\
        &= \prod_{m\in\mathbb{Z}_d}\left[\left(1-u\omega_d^{-m}h_v\right)T_{G\setminus\mathcal{N}[v]}\left(u\omega_d^{-m}\right)-u\omega_d^{-m}\sum_{w\in\mathcal{N}(v)}h_wT_{G\setminus\mathcal{N}[w]}\left(u\omega_d^{-m}\right)\right]. \notag
    \end{align}
    For a subset of vertices $U$ of $G$ and $\alpha\in\mathbb{N}^U$, define $b^\alpha\coloneqq\prod_{w \in U}b_w^{\alpha_w}$. By expanding the product and using $h_wh_v=\omega_dh_vh_w$ for all $w\in\mathcal{N}(v)$, we have
    \begin{equation}
        \prod_{m\in\mathbb{Z}_d}T_G\left(u\omega_d^{-m}\right) = \sum_{l=0}^d\left[\prod_{k=0}^{l-1}\left(1-u\omega_d^{-k}h_v\right)\right]\sum_{\substack{\alpha\in\mathbb{N}^{\mathcal{N}(v)} \\ \abs{\alpha}=d-l}}b^\alpha[b^\alpha]Z_{G\setminus\{v\}}\left(-u^d\right), \notag
    \end{equation}
    where $[b^\alpha]Z_{G\setminus\{v\}}\left(-u^d\right)$ denotes the $b^\alpha$ coefficient in $Z_{G\setminus\{v\}}\left(-u^d\right)$. Hence, by applying Proposition~\ref{proposition:CliqueRecursionRelation},
    \begin{align}
        \prod_{m\in\mathbb{Z}_d}T_G\left(u\omega_d^{-m}\right) &= \left[\prod_{k=0}^{d-1}\left(1-u\omega_d^{-k}h_v\right)\right]Z_{G\setminus\mathcal{N}[v]}\left(-u^d\right)
        -u^d\sum_{w\in\mathcal{N}(v)}b_w^dZ_{G\setminus\mathcal{N}[w]}\left(-u^d\right) \notag \\
        &= \left[1-(ub_v)^d\right]Z_{G\setminus\mathcal{N}[v]}\left(-u^d\right)
        -u^d\sum_{w\in\mathcal{N}(v)}b_w^dZ_{G\setminus\mathcal{N}[w]}\left(-u^d\right) \notag \\
        &= Z_{G\setminus\mathcal{N}[v]}\left(-u^d\right)
        -u^d\sum_{w\in\mathcal{N}[v]}b_w^dZ_{G\setminus\mathcal{N}[w]}\left(-u^d\right) \notag \\
        &= Z_G\left(-u^d\right). \notag
    \end{align}
    This completes the proof.
\end{proof}

\section{Proof of Lemma~\ref*{lemma:HamiltonianParafermionicModeCommutator}}
\label{section:HamiltonianParafermionicModeCommutator}

\HamiltonianParafermionicModeCommutator*

To prove Lemma~\ref{lemma:HamiltonianParafermionicModeCommutator}, we first require the following lemmas.

\begin{lemma}
    \label{lemma:TransferOperatorSimplicialModeCommutator}
    Fix $d\in\mathbb{Z}_{\geq2}$. Let $H$ be a qudit Hamiltonian with oriented indifference frustration graph $G$. Further let $\chi$ be a simplicial mode with respect to an oriented perfect elimination ordering of $G$ and let $v$ denote the last vertex in the oriented perfect elimination ordering of $G$. Then, for all $p\in\mathbb{Z}$,
    \begin{equation}
        \frac{\omega_d^{p+1}}{u}\left[T_G\left(u\omega_d^{-p}\right),\chi\right] = (1-\omega_d)h_v \chi T_{G \setminus N[v]}\left(u\omega_d^{-p}\right). \notag
    \end{equation}
\end{lemma}

\begin{proof}
    By applying Proposition~\ref{proposition:VertexRecursionRelation} with $h_v\chi=\omega_d\chi h_v$,
    \begin{align}
        \frac{\omega_d^{p+1}}{u}\left[T_G\left(u\omega_d^{-p}\right),\chi\right] &= \omega_d\left[\chi h_v-h_v \chi\right]T_{G\setminus\mathcal{N}[v]}\left(u\omega_d^{-p}\right) \notag \\
        &= (1-\omega_d)h_v\chi T_{G\setminus\mathcal{N}[v]}\left(u\omega_d^{-p}\right), \notag
    \end{align}
    completing the proof.
\end{proof}

\begin{lemma}
    \label{lemma:TransferOperatorHamiltonianSimplicialModeCommutator}
    Fix $d\in\mathbb{Z}_{\geq2}$. Let $H$ be a qudit Hamiltonian with oriented indifference frustration graph $G$. Further let $\chi$ be a simplicial mode with respect to an oriented perfect elimination ordering of $G$. Then, for all $p\in\mathbb{Z}$,
    \begin{equation}
        T_G\left(u\omega_d^{-p}\right)[H,\chi]-\omega_d[H,\chi]T_G\left(u\omega_d^{-p}\right) = (1-\omega_d)\frac{\omega_d^{p+1}}{u}\left[T_G\left(u\omega_d^{-p}\right),\chi\right]. \notag
    \end{equation}
\end{lemma}

\begin{proof}
    Let $v$ denote the last vertex in the oriented perfect elimination ordering of $G$. By using $[H,\chi]=(1-\omega_d)h_v\chi$,
    \begin{equation}
        T_G\left(u\omega_d^{-p}\right)[H,\chi]-\omega_d[H,\chi]T_G\left(u\omega_d^{-p}\right) = (1-\omega_d)\left[T_G\left(u\omega_d^{-p}\right)h_v\chi-\omega_dh_v \chi T_G\left(u\omega_d^{-p}\right)\right]. \notag
    \end{equation}
    By applying Proposition~\ref{proposition:VertexRecursionRelation} with $h_v\chi=\omega_d\chi h_v$,
    \begin{align}
        &T_G\left(u\omega_d^{-p}\right)[H,\chi]-\omega_d[H,\chi]T_G\left(u\omega_d^{-p}\right) \notag \\
        &\quad= (1-\omega_d)\left[\left(T_{G\setminus\{v\}}\left(u\omega_d^{-p}\right)h_v-\omega_dh_vT_{G\setminus\{v\}}\left(u\omega_d^{-p}\right)\right)\chi\right] \notag \\
        &\quad\quad+ (1-\omega_d)\left[u\omega_d^{-p}h_v\left(\omega_d \chi h_v-h_v \chi\right)T_{G\setminus\mathcal{N}[v]}\left(u\omega_d^{-p}\right)\right] \notag \\
        &\quad= (1-\omega_d)\left[T_{G\setminus\{v\}}\left(u\omega_d^{-p}\right)h_v-\omega_dh_vT_{G\setminus\{v\}}\left(u\omega_d^{-p}\right)\right]\chi. \notag
    \end{align}
    Note that since $v$ is the last vertex in an oriented perfect elimination ordering of $G$, $N(v)$ is a clique. By applying Proposition~\ref{proposition:CliqueRecursionRelation} with $h_wh_v=\omega_dh_vh_w$ for all $w\in\mathcal{N}(v)$,
    \begin{align}
        &T_G\left(u\omega_d^{-p}\right)[H,\chi]-\omega_d[H,\chi]T_G\left(u\omega_d^{-p}\right) \notag \\
        &\quad= (1-\omega_d)\left[(1-\omega_d)h_vT_{G \setminus N[v]}\left(u\omega_d^{-p}\right)+u\omega_d^{-p}\sum_{w \in N(v)}\left(\omega_dh_vh_w-h_wh_v\right)T_{G\setminus\mathcal{N}[w]}\left(u\omega_d^{-p}\right)\right]\chi \notag \\
        &\quad= (1-\omega_d)^2h_v \chi T_{G \setminus N[v]}\left(u\omega_d^{-p}\right). \notag
    \end{align}
    The proof then follows from Lemma~\ref{lemma:TransferOperatorSimplicialModeCommutator}.
\end{proof}

\begin{proof}[Proof of Lemma~\ref*{lemma:HamiltonianParafermionicModeCommutator}.]
    By applying Lemma~\ref{lemma:TransferOperatorHamiltonianSimplicialModeCommutator} and Lemma~\ref{lemma:IndependencePolynomialTransferMatrices} and using Corollary~\ref{corollary:HamiltonianTransferOperatorCommutator}, we obtain
    \begin{align}
        [H,\psi_{p,k}] &= \frac{1}{N_k}T_G^p\left(\varepsilon_k^{-1}\right)[H,\chi]T_G^{\succ p}\left(\varepsilon_k^{-1}\right) \notag \\
        &= \frac{1}{N_k}\left[\omega_d[H,\chi]T_G^p\left(\varepsilon_k^{-1}\right)+(1-\omega_d)\omega_d^{p+1}\varepsilon_k\left[T_G^p\left(\varepsilon_k^{-1}\right),\chi\right]\right]T_G^{\succ p}\left(\varepsilon_k^{-1}\right) \notag \\
        &= \frac{1}{N_k}\left[(1-\omega_d)\omega_d^{p+1}\varepsilon_kT_G^p\left(\varepsilon_k^{-1}\right) \chi T_G^{\succ p}\left(\varepsilon_k^{-1}\right)\right] \notag \\
        &\quad+ \frac{1}{N_k}\left[\omega_d\left([H,\chi]-(1-\omega_d)\omega_d^p\varepsilon_k\chi\right)T_G^p\left(\varepsilon_k^{-1}\right)T_G^{\succ p}\left(\varepsilon_k^{-1}\right)\right] \notag \\
        &= (1-\omega_d)\omega_d^{p+1}\varepsilon_k\psi_{p,k}+\frac{\omega_d}{N_k}\left([H,\chi]-(1-\omega_d)\omega_d^p\varepsilon_k\chi\right)Z_G\left(-\varepsilon_k^{-d}\right). \notag
    \end{align}
    Now, since $Z_G\left(-\varepsilon_k^{-d}\right)=0$,
    \begin{equation}
         [H,\psi_{p,k}] = (1-\omega_d)\omega_d^{p+1}\varepsilon_k\psi_{p,k}, \notag
    \end{equation}
    completing the proof.
\end{proof}

\section{Proof of Lemma~\ref*{lemma:ParafermionicModeCommutator}}
\label{section:ParafermionicModeCommutator}

\ParafermionicModeCommutator*

To prove Lemma~\ref{lemma:ParafermionicModeCommutator}, we first require the following lemmas.

\begin{lemma}
    \label{lemma:TransferOperatorParafermionicModeCommutator}
    Fix $d\in\mathbb{Z}_{\geq2}$. Let $H$ be a qudit Hamiltonian with oriented indifference frustration graph $G$. Then, for all $p\in\mathbb{Z}$, the single-particle energies $\{\varepsilon_l\}$ and parafermionic modes $\{\psi_{q,l}\}$ satisfy
    \begin{equation}
        \left(1-u\omega_d^{q-p}\varepsilon_l\right)T_G\left(u\omega_d^{-p}\right)\psi_{q,l} = \left(1-u\omega_d^{q-p+1}\varepsilon_l\right)\psi_{q,l}T_G\left(u\omega_d^{-p}\right). \notag
    \end{equation}
\end{lemma}

\begin{proof}
    By applying Lemma~\ref{lemma:TransferOperatorHamiltonianSimplicialModeCommutator} and using Corollary~\ref{corollary:HamiltonianTransferOperatorCommutator}, we obtain
    \begin{align}
        &(1-\omega_d)\omega_d^{p+1}\left[T_G\left(u\omega_d^{-p}\right),\psi_{q,l}\right] \notag \\
        &\quad= \frac{(1-\omega_d)\omega_d^{p+1}}{N_l}T_G^q\left(\varepsilon_l^{-1}\right)\left[T_G\left(u\omega_d^{-p}\right),\chi\right]T_G^{\succ q}\left(\varepsilon_l^{-1}\right) \notag \\
        &\quad= \frac{u}{N_l}T_G^q\left(\varepsilon_l^{-1}\right)\left[T_G\left(u\omega_d^{-p}\right)[H,\chi]-\omega_d[H,\chi]T_G\left(u\omega_d^{-p}\right)\right]T_G^{\succ q}\left(\varepsilon_l^{-1}\right) \notag \\
        &\quad= u\left[T_G\left(u\omega_d^{-p}\right)[H,\psi_{q,l}]-\omega_d[H,\psi_{q,l}]T_G\left(u\omega_d^{-p}\right)\right]. \notag
    \end{align}
    Now, by applying Lemma~\ref{lemma:HamiltonianParafermionicModeCommutator},
    \begin{equation}
        \left[T_G\left(u\omega_d^{-p}\right),\psi_{q,l}\right] = u\omega_d^{q-p}\varepsilon_l\left[T_G\left(u\omega_d^{-p}\right)\psi_{q,l}-\omega_d\psi_{q,l}T_G\left(u\omega_d^{-p}\right)\right]. \notag
    \end{equation}
    Therefore,
    \begin{equation}
        \left(1-u\omega_d^{q-p}\varepsilon_l\right)T_G\left(u\omega_d^{-p}\right)\psi_{q,l} = \left(1-u\omega_d^{q-p+1}\varepsilon_l\right)\psi_{q,l}T_G\left(u\omega_d^{-p}\right), \notag
    \end{equation}
    completing the proof.
\end{proof}

\begin{lemma}
    \label{lemma:SimplicialModeParafermionicModeCommutator}
    Fix $d\in\mathbb{Z}_{\geq2}$. Let $H$ be a qudit Hamiltonian with oriented indifference frustration graph $G$. Further let $\chi$ be a simplicial mode with respect to an oriented perfect elimination ordering of $G$. The single-particle energies $\{\varepsilon_l\}$ and parafermionic modes $\{\psi_{q,l}\}$ satisfy
    \begin{equation}
        (1+\omega_d)\chi\psi_{q,l} = \frac{1}{N_l}T_G^q\left(\varepsilon_l^{-1}\right)\chi^2T_G^{\succ q}\left(\varepsilon_l^{-1}\right). \notag
    \end{equation}
\end{lemma}

\begin{proof}
    Let $v$ denote the last vertex in the oriented perfect elimination ordering of $G$. By Proposition~\ref{proposition:VertexRecursionRelation},
    \begin{align}
        (1+\omega_d)\chi\psi_{q,l} &= \frac{(1+\omega_d)}{N_l}\chi T_G^q\left(\varepsilon_l^{-1}\right) \chi T_G^{\succ q}\left(\varepsilon_l^{-1}\right) \notag \\
        &= \frac{(1+\omega_d)}{N_l}\chi\left[T_{G\setminus\{v\}}^q\left(\varepsilon_l^{-1}\right)-\omega_d^{-q}\varepsilon_l^{-1}h_vT_{G\setminus\mathcal{N}[v]}^q\left(\varepsilon_l^{-1}\right)\right] \chi T_G^{\succ q}\left(\varepsilon_l^{-1}\right) \notag \\
        &= \frac{1}{N_l}\chi\left[(1+\omega_d)T_{G\setminus\{v\}}^q\left(\varepsilon_l^{-1}\right)-\omega_d^{-q+1}\left(1+\omega_d^{-1}\right)\varepsilon_l^{-1}h_vT_{G\setminus\mathcal{N}[v]}^q\left(\varepsilon_l^{-1}\right)\right] \chi T_G^{\succ q}\left(\varepsilon_l^{-1}\right) \notag \\
        &= \frac{1}{N_l}\chi\left[T_{G\setminus\{v\}}^q\left(\varepsilon_l^{-1}\right)-\omega_d^{-q+1}\varepsilon_l^{-1}h_vT_{G\setminus\mathcal{N}[v]}^q\left(\varepsilon_l^{-1}\right)\right] \chi T_G^{\succ q}\left(\varepsilon_l^{-1}\right) \notag \\
        &\quad+ \frac{1}{N_l}\chi\omega_d\left[T_{G\setminus\{v\}}^q\left(\varepsilon_l^{-1}\right)-\omega_d^{-q-1}\varepsilon_l^{-1}h_vT_{G\setminus\mathcal{N}[v]}^q\left(\varepsilon_l^{-1}\right)\right] \chi T_G^{\succ q}\left(\varepsilon_l^{-1}\right). \notag
    \end{align}
    By Proposition~\ref{proposition:VertexRecursionRelation} and Lemma~\ref{lemma:IndependencePolynomialTransferMatrices} with $h_v\chi=\omega_d\chi h_v$, we have
    \begin{align}
        (1+\omega_d)\chi\psi_{q,l} &= \frac{1}{N_l}\left[T_{G\setminus\{v\}}^q\left(\varepsilon_l^{-1}\right)-\omega_d^{-q}\varepsilon_l^{-1}h_vT_{G\setminus\mathcal{N}[v]}^q\left(\varepsilon_l^{-1}\right)\right] \chi^2 T_G^{\succ q}\left(\varepsilon_l^{-1}\right) \notag \\
        &\quad+ \frac{1}{N_l}\chi^2\omega_d\left[T_{G\setminus\{v\}}^q\left(\varepsilon_l^{-1}\right)-\omega_d^{-q}\varepsilon_l^{-1}h_vT_{G\setminus\mathcal{N}[v]}^q\left(\varepsilon_l^{-1}\right)\right]T_G^{\succ q}\left(\varepsilon_l^{-1}\right) \notag \\
        &= \frac{1}{N_l}\left[T_G^q\left(\varepsilon_l^{-1}\right)\chi^2+\omega_d\chi^2T_G^q\left(\varepsilon_l^{-1}\right)\right]T_G^{\succ q}\left(\varepsilon_l^{-1}\right) \notag \\
        &= \frac{1}{N_l}\left[T_G^q\left(\varepsilon_l^{-1}\right)\chi^2T_G^{\succ q}\left(\varepsilon_l^{-1}\right)+\omega_d\chi^2Z_G\left(-\varepsilon_l^{-d}\right)\right]. \notag
    \end{align}
    Now, since $Z_G\left(-\varepsilon_l^{-d}\right)=0$,
    \begin{equation}
        (1+\omega_d)\chi\psi_{q,l} = \frac{1}{N_l}T_G^q\left(\varepsilon_l^{-1}\right)\chi^2T_G^{\succ q}\left(\varepsilon_l^{-1}\right), \notag
    \end{equation}
    completing the proof.
\end{proof}

\begin{proof}[Proof of Lemma~\ref*{lemma:ParafermionicModeCommutator}.]
    By applying Lemma~\ref{lemma:TransferOperatorParafermionicModeCommutator},
    \begin{align}
        \psi_{p,k}\psi_{q,l} &= \frac{1}{N_k}T_G^p\left(\varepsilon_k^{-1}\right) \chi T_G^{\succ p}\left(\varepsilon_k^{-1}\right)\psi_{q,l} \notag \\
        &= \frac{1}{N_k}\left(\prod_{m=1}^{d-1}\frac{1-\omega_d^{q-p-m+1}\varepsilon_k^{-1}\varepsilon_l}{1-\omega_d^{q-p-m}\varepsilon_k^{-1}\varepsilon_l}\right)T_G^p\left(\varepsilon_k^{-1}\right)\chi\psi_{q,l}T_G^{\succ p}\left(\varepsilon_k^{-1}\right) \notag \\
        &= \frac{1}{N_k}\left(\frac{\omega_d^p\varepsilon_k-\omega_d^q\varepsilon_l}{\omega_d^p\varepsilon_k-\omega_d^{q+1}\varepsilon_l}\right)T_G^p\left(\varepsilon_k^{-1}\right)\chi\psi_{q,l}T_G^{\succ p}\left(\varepsilon_k^{-1}\right), \notag
    \end{align}
    for $k$ and $l$ distinct. Now, by applying Lemma~\ref{lemma:SimplicialModeParafermionicModeCommutator} and using Corollary~\ref{corollary:HamiltonianTransferOperatorCommutator}, we have
    \begin{equation}
        \psi_{p,k}\psi_{q,l} = \frac{1}{N_kN_l}\frac{1}{\left(1+\omega_d\right)}\left(\frac{\omega^p\varepsilon_k-\omega_d^q\varepsilon_l}{\omega_d^p\varepsilon_k-\omega_d^{q+1}\varepsilon_l}\right)T_G^p\left(\varepsilon_k^{-1}\right)T_G^q\left(\varepsilon_l^{-1}\right)\chi^2T_G^{\succ p}\left(\varepsilon_k^{-1}\right)T_G^{\succ q}\left(\varepsilon_l^{-1}\right). \notag
    \end{equation}
    Similarly,
    \begin{equation}
        \psi_{q,l}\psi_{p,k} = \frac{1}{N_kN_l}\frac{1}{\left(1+\omega_d\right)}\left(\frac{\omega_d^q\varepsilon_l-\omega_d^p\varepsilon_k}{\omega_d^q\varepsilon_l-\omega_d^{p+1}\varepsilon_k}\right)T_G^p\left(\varepsilon_k^{-1}\right)T_G^q\left(\varepsilon_l^{-1}\right)\chi^2T_G^{\succ p}\left(\varepsilon_k^{-1}\right)T_G^{\succ q}\left(\varepsilon_l^{-1}\right). \notag
    \end{equation}
    Therefore,
    \begin{equation}
        \left(\omega_d^p\varepsilon_k-\omega_d^{q+1}\varepsilon_l\right)\psi_{p,k}\psi_{q,l} = -\left(\omega_d^q\varepsilon_l-\omega_d^{p+1}\varepsilon_k\right)\psi_{q,l}\psi_{p,k}, \notag
    \end{equation}
    for $k$ and $l$ distinct. By applying Lemma~\ref{lemma:TransferOperatorParafermionicModeCommutator} and using Corollary~\ref{corollary:HamiltonianTransferOperatorCommutator}, we obtain
    \begin{equation}
        \psi_{p,k}\psi_{q,k} = 0, \notag
    \end{equation}
    for $p$ and $q+1$ distinct. This completes the proof.
\end{proof}

\section{Proof of Lemma~\ref*{lemma:ProjectorOperatorRelations}}
\label{section:ProjectorOperatorRelations}

\ProjectorOperatorRelations*

To prove Lemma~\ref{lemma:ProjectorOperatorRelations}, we first require the following lemmas.

\begin{lemma}
    \label{lemma:TransferOperatorParafermionicModeEqualCommutator}
    Fix $d\in\mathbb{Z}_{\geq2}$. Let $H$ be a qudit Hamiltonian with oriented indifference frustration graph $G$. Then the single-particle energies $\{\varepsilon_k\}$ and parafermionic modes $\{\psi_{p,k}\}$ satisfy
    \begin{equation}
        T_G^p\left(\varepsilon_k^{-1}\right)\psi_{p,k} = (1-\omega_d)\varepsilon_k\psi_{p,k}\frac{\partial T_G^p\left(\varepsilon^{-1}\right)}{\partial\varepsilon}\bigg|_{\varepsilon=\varepsilon_k}. \notag 
    \end{equation}
\end{lemma}

\begin{proof}
    By applying Lemma~\ref{lemma:TransferOperatorParafermionicModeCommutator},
    \begin{align}
        T_G^p\left(\varepsilon_k^{-1}\right)\psi_{p,k} &= \lim_{\varepsilon\to\varepsilon_k}T_G^p\left(\varepsilon^{-1}\right)\psi_{p,k} \notag \\
        &= \lim_{\varepsilon\to\varepsilon_k}\left(\frac{\varepsilon-\omega_d\varepsilon_k}{\varepsilon-\varepsilon_k}\right)\psi_{p,k}T_G^p\left(\varepsilon^{-1}\right). \notag
    \end{align}
    Now, by applying Lemma~\ref{lemma:IndependencePolynomialTransferMatrices} and using $Z_G\left(-\varepsilon_l^{-d}\right)=0$,
    \begin{equation}
        T_G^p\left(\varepsilon_k^{-1}\right)\psi_{p,k} = (1-\omega_d)\varepsilon_k\psi_{p,k}\frac{\partial T_G^p\left(\varepsilon^{-1}\right)}{\partial\varepsilon}\bigg|_{\varepsilon=\varepsilon_k}, \notag
    \end{equation}
    completing the proof.
\end{proof}

\begin{lemma}
    \label{lemma:ProjectorAlgebraicIdentity}
    Fix $d\in\mathbb{Z}^+$. Then,
    \begin{equation}
        \prod_{p=0}^{d-2}\left(1-\omega_d^{p+1}\right) = d. \notag
    \end{equation}
\end{lemma}

\begin{proof}
    First observe that the polynomials
    \begin{equation}
        \prod_{p=0}^{d-1}\left(1-u\omega_d^p\right) \quad\text{and}\quad 1-u^d \notag  
    \end{equation}
    are equal since they have identical roots and identical constant term. By considering the derivative, we obtain
    \begin{equation}
        \prod_{p=0}^{d-2}\left(1-u\omega_d^{p+1}\right)-(1-u)\frac{\partial}{\partial u}\prod_{p=0}^{d-2}\left(1-u\omega_d^{p+1}\right) = du^{d-1}. \notag  
    \end{equation}
    Evaluating this expression at $u=1$,
    \begin{equation}
        \prod_{p=0}^{d-2}\left(1-\omega_d^{p+1}\right) = d, \notag  
    \end{equation}
    completing the proof.
\end{proof}

\begin{lemma}
    \label{lemma:TransferOperatorSimplicalModeParafermionicModeRelation}
    Fix $d\in\mathbb{Z}_{\geq2}$. Let $H$ be a qudit Hamiltonian with oriented indifference frustration graph $G$. Further let $\chi$ be a simplicial mode with respect to an oriented perfect elimination ordering of $G$ and let $v$ denote the last vertex in the oriented perfect elimination ordering of $G$. Then the single-particle energies $\{\varepsilon_k\}$ and parafermionic modes $\{\psi_{p,k}\}$ satisfy
    \begin{equation}
        T_G^{p+1}\left(\varepsilon_k^{-1}\right)\chi\psi_{p,k} = (1-\omega_d)\omega_d^{-p-2}\varepsilon_k^{-1}h_v \chi T_{G \setminus N[v]}^{p+1}\left(\varepsilon_k^{-1}\right)\psi_{p,k}. \notag
    \end{equation}
\end{lemma}

\begin{proof}
    By applying Lemma~\ref{lemma:TransferOperatorSimplicialModeCommutator},
    \begin{equation}
        T_G^{p+1}\left(\varepsilon_k^{-1}\right)\chi\psi_{p,k} = \left[\chi T_G^{p+1}\left(\varepsilon_k^{-1}\right)+(1-\omega_d)\omega_d^{-p-2}\varepsilon_k^{-1}h_v \chi T_{G \setminus N[v]}^{p+1}\left(\varepsilon_k^{-1}\right)\right]\psi_{p,k}. \notag
    \end{equation}
    Now, by applying Lemma~\ref{lemma:TransferOperatorParafermionicModeCommutator},
    \begin{equation}
        T_G^{p+1}\left(\varepsilon_k^{-1}\right)\chi\psi_{p,k} = (1-\omega_d)\omega_d^{-p-2}\varepsilon_k^{-1}h_v \chi T_{G \setminus N[v]}^{p+1}\left(\varepsilon_k^{-1}\right)\psi_{p,k}, \notag
    \end{equation}
    completing the proof.
\end{proof}

\begin{lemma}
    \label{lemma:TransferOperatorDerivativeParafermionicModeRelation}
    Fix $d\in\mathbb{Z}_{\geq2}$. Let $H$ be a qudit Hamiltonian with oriented indifference frustration graph $G$. Then the single-particle energies $\{\varepsilon_k\}$ and parafermionic modes $\{\psi_{p,k}\}$ satisfy
    \begin{equation}
        T_G^{\succ p+1}\left(\varepsilon_k^{-1}\right)\frac{\partial T_G^{p+1}\left(\varepsilon^{-1}\right)}{\partial\varepsilon}\bigg|_{\varepsilon=\varepsilon_k}\psi_{p,k} = \frac{\partial Z_G\left(-\varepsilon^{-d}\right)}{\partial\varepsilon}\bigg|_{\varepsilon=\varepsilon_k}\psi_{p,k}. \notag 
    \end{equation}
\end{lemma}

\begin{proof}
    It follows from Lemma~\ref{lemma:IndependencePolynomialTransferMatrices} that
    \begin{equation}
        T_G^{\succ p+1}\left(\varepsilon_k^{-1}\right)\frac{\partial T_G^{p+1}\left(\varepsilon^{-1}\right)}{\partial\varepsilon}\bigg|_{\varepsilon=\varepsilon_k}+\frac{\partial T_G^{\succ p+1}\left(\varepsilon^{-1}\right)}{\partial\varepsilon}\bigg|_{\varepsilon=\varepsilon_k}T_G^{p+1}\left(\varepsilon_k^{-1}\right) = \frac{\partial Z_G\left(-\varepsilon^{-d}\right)}{\partial\varepsilon}\bigg|_{\varepsilon=\varepsilon_k}. \notag
    \end{equation}
    Hence, by applying Lemma~\ref{lemma:TransferOperatorParafermionicModeCommutator},
    \begin{align}
        T_G^{\succ p+1}\left(\varepsilon_k^{-1}\right)\frac{\partial T_G^{p+1}\left(\varepsilon^{-1}\right)}{\partial\varepsilon}\bigg|_{\varepsilon=\varepsilon_k}\psi_{p,k} &= \left[\frac{\partial Z_G\left(-\varepsilon^{-d}\right)}{\partial\varepsilon}\bigg|_{\varepsilon=\varepsilon_k}-\frac{\partial T_G^{\succ p+1}\left(\varepsilon^{-1}\right)}{\partial\varepsilon}\bigg|_{\varepsilon=\varepsilon_k}T_G^{p+1}\left(\varepsilon_k^{-1}\right)\right]\psi_{p,k} \notag \\
        &= \frac{\partial Z_G\left(-\varepsilon^{-d}\right)}{\partial\varepsilon}\bigg|_{\varepsilon=\varepsilon_k}\psi_{p,k}, \notag
    \end{align}
    completing the proof.
\end{proof}

\begin{lemma}
    \label{lemma:ParafermionicModeVertexDeletionRelation}
    Fix $d\in\mathbb{Z}_{\geq2}$. Let $H$ be a qudit Hamiltonian with oriented indifference frustration graph $G$. Further let $\chi$ be a simplicial mode with respect to an oriented perfect elimination ordering of $G$ and let $v$ denote the last vertex in the oriented perfect elimination ordering of $G$. Then the single-particle energies $\{\varepsilon_k\}$ and parafermionic modes $\{\psi_{p,k}\}$ satisfy
    \begin{equation}
        \psi_{p,k} = \frac{1}{N_k}(1-\omega_d)\omega_d^{-p-1}\varepsilon_k^{-1}h_v \chi T_{G\setminus\mathcal{N}[v]}^p\left(\varepsilon_k^{-1}\right)T_G^{\succ p}\left(\varepsilon_k^{-1}\right). \notag
    \end{equation}
\end{lemma}

\begin{proof}
    By applying Lemma~\ref{lemma:TransferOperatorSimplicialModeCommutator},
    \begin{align}
        \psi_{p,k} &= \frac{1}{N_k}T_G^p\left(\varepsilon_k^{-1}\right) \chi T_G^{\succ p}\left(\varepsilon_k^{-1}\right) \notag \\
        &= \frac{1}{N_k}\left[ \chi T_G^p\left(\varepsilon_k^{-1}\right)+(1-\omega_d)\omega_d^{-p-1}\varepsilon_k^{-1}h_v \chi T_{G\setminus\mathcal{N}[v]}^p\left(\varepsilon_k^{-1}\right)\right]T_G^{\succ p}\left(\varepsilon_k^{-1}\right). \notag
    \end{align}
    Now, by applying Lemma~\ref{lemma:IndependencePolynomialTransferMatrices} and using $Z_G\left(-\varepsilon_k^{-d}\right)=0$,
    \begin{equation}
        \psi_{p,k} = \frac{1}{N_k}(1-\omega_d)\omega_d^{-p-1}\varepsilon_k^{-1}h_v \chi T_{G\setminus\mathcal{N}[v]}^p\left(\varepsilon_k^{-1}\right)T_G^{\succ p}\left(\varepsilon_k^{-1}\right). \notag
    \end{equation}
    This completes the proof.
\end{proof}

\begin{lemma}
    \label{lemma:TransferOperatorDerivativeIdempotenceRelation}
    Fix $d\in\mathbb{Z}_{\geq2}$. Let $H$ be a qudit Hamiltonian with oriented indifference frustration graph $G$. Then, for all $p\in\mathbb{Z}$, the single-particle energies $\{\varepsilon_l\}$ and parafermionic modes $\{\psi_{q,l}\}$ satisfy
    \begin{equation}
        \left[T_G^{\succ p}\left(\varepsilon_k^{-1}\right)\frac{\partial T_G^p\left(\varepsilon^{-1}\right)}{\partial\varepsilon}\bigg|_{\varepsilon=\varepsilon_k}\right]^2 = \frac{\partial Z_G\left(-\varepsilon^{-d}\right)}{\partial\varepsilon}\bigg|_{\varepsilon=\varepsilon_k}T_G^{\succ p}\left(\varepsilon_k^{-1}\right)\frac{\partial T_G^p\left(\varepsilon^{-1}\right)}{\partial\varepsilon}\bigg|_{\varepsilon=\varepsilon_k}. \notag
    \end{equation}
\end{lemma}

\begin{proof}
    It follows from Lemma~\ref{lemma:IndependencePolynomialTransferMatrices} that
    \begin{equation}
        T_G^{\succ p+1}\left(\varepsilon_k^{-1}\right)\frac{\partial T_G^{p+1}\left(\varepsilon^{-1}\right)}{\partial\varepsilon}\bigg|_{\varepsilon=\varepsilon_k}+\frac{\partial T_G^{\succ p+1}\left(\varepsilon^{-1}\right)}{\partial\varepsilon}\bigg|_{\varepsilon=\varepsilon_k}T_G^{p+1}\left(\varepsilon_k^{-1}\right) = \frac{\partial Z_G\left(-\varepsilon^{-d}\right)}{\partial\varepsilon}\bigg|_{\varepsilon=\varepsilon_k}. \notag
    \end{equation}
    Hence, by applying Lemma~\ref{lemma:IndependencePolynomialTransferMatrices} and using $Z_G\left(-\varepsilon_k^{-d}\right)=0$, we obtain
    \begin{align}
        &\left[T_G^{\succ p}\left(\varepsilon_k^{-1}\right)\frac{\partial T_G^p\left(\varepsilon^{-1}\right)}{\partial\varepsilon}\bigg|_{\varepsilon=\varepsilon_k}\right]^2 \notag \\ 
        &\quad= T_G^{\succ p}\left(\varepsilon_k^{-1}\right)\left[\frac{\partial Z_G\left(-\varepsilon^{-d}\right)}{\partial\varepsilon}\bigg|_{\varepsilon=\varepsilon_k}-T_G^{\succ p+1}\left(\varepsilon_k^{-1}\right)\frac{\partial T_G^{p+1}\left(\varepsilon^{-1}\right)}{\partial\varepsilon}\bigg|_{\varepsilon=\varepsilon_k}\right]\frac{\partial T_G^p\left(\varepsilon^{-1}\right)}{\partial\varepsilon}\bigg|_{\varepsilon=\varepsilon_k} \notag \\
        &= \frac{\partial Z_G\left(-\varepsilon^{-d}\right)}{\partial\varepsilon}\bigg|_{\varepsilon=\varepsilon_k}T_G^{\succ p}\left(\varepsilon_k^{-1}\right)\frac{\partial T_G^p\left(\varepsilon^{-1}\right)}{\partial\varepsilon}\bigg|_{\varepsilon=\varepsilon_k}. \notag
    \end{align}
    This completes the proof.
\end{proof}

\begin{proof}[Proof of Lemma~\ref*{lemma:ProjectorOperatorRelations}.]
    By applying Lemma~\ref{lemma:TransferOperatorParafermionicModeCommutator} and Lemma~\ref{lemma:TransferOperatorParafermionicModeEqualCommutator} successively and then by applying Lemma~\ref{lemma:ProjectorAlgebraicIdentity}, we obtain
    \begin{align}
        \mathcal{P}_{r,k} &= \frac{1}{N_k}T_G^{r-1}\left(\varepsilon_k^{-1}\right) \chi T_G^{\succ r-1}\left(\varepsilon_k^{-1}\right)\prod_{p=2}^d\psi_{r-p,k} \notag \\
        &= \frac{1}{N_k}T_G^{r-1}\left(\varepsilon_k^{-1}\right) \chi\left[\prod_{p=2}^d(1-\omega_d)\varepsilon_k\left(\prod_{m=0}^{d-p-1}\frac{1-\omega_d^{-p-m+1}}{1-\omega_d^{-p-m}}\right)\psi_{r-p,k}\frac{\partial T_G^{r-p}\left(\varepsilon^{-1}\right)}{\partial\varepsilon}\bigg|_{\varepsilon=\varepsilon_k}\right] \notag \\
        &= \frac{1}{N_k}T_G^{r-1}\left(\varepsilon_k^{-1}\right) \chi\left[\prod_{p=2}^d\left(1-\omega_d^{-p+1}\right)\varepsilon_k\psi_{r-p,k}\frac{\partial T_G^{r-p}\left(\varepsilon^{-1}\right)}{\partial\varepsilon}\bigg|_{\varepsilon=\varepsilon_k}\right] \notag \\
        &= \frac{d\varepsilon_k^{d-1}}{N_k}T_G^{r-1}\left(\varepsilon_k^{-1}\right) \chi\left[\prod_{p=2}^d\psi_{r-p,k}\frac{\partial T_G^{r-p}\left(\varepsilon^{-1}\right)}{\partial\varepsilon}\bigg|_{\varepsilon=\varepsilon_k}\right]. \notag
    \end{align}
    Let $v$ denote the last vertex in the oriented perfect elimination ordering of $G$. By applying Lemma~\ref{lemma:TransferOperatorSimplicalModeParafermionicModeRelation} and Lemma~\ref{lemma:TransferOperatorDerivativeParafermionicModeRelation} successively and then by applying Lemma~\ref{lemma:ParafermionicModeVertexDeletionRelation}, we have
    \begin{align}
        \mathcal{P}_{r,k} &= \frac{d(1-\omega_d)^{d-1}}{N_k^{d-1}}\left[\frac{\partial Z_G\left(-\varepsilon^{-d}\right)}{\partial\varepsilon}\bigg|_{\varepsilon=\varepsilon_k}\right]^{d-2}\left[\prod_{p=1}^{d-1}\omega_d^{p-r-1}h_v \chi T_{G\setminus\mathcal{N}[v]}^{r-p}\left(\varepsilon_k^{-1}\right)\right]\psi_{r,k}\frac{\partial T_G^r\left(\varepsilon^{-1}\right)}{\partial\varepsilon}\bigg|_{\varepsilon=\varepsilon_k} \notag \\
        &= \frac{d(1-\omega_d)^d\varepsilon_k^{-1}}{N_k^d}\left[\frac{\partial Z_G\left(-\varepsilon^{-d}\right)}{\partial\varepsilon}\bigg|_{\varepsilon=\varepsilon_k}\right]^{d-2}\left[\prod_{p=1}^d\omega_d^{p-1}h_v \chi T_{G\setminus\mathcal{N}[v]}^{r-p}\left(\varepsilon_k^{-1}\right)\right] \notag \\ 
        &\quad\times T_G^{\succ r}\left(\varepsilon_k^{-1}\right)\frac{\partial T_G^r\left(\varepsilon^{-1}\right)}{\partial\varepsilon}\bigg|_{\varepsilon=\varepsilon_k}. \notag
    \end{align}
    By applying Lemma~\ref{lemma:IndependencePolynomialTransferMatrices} with $h_v\chi=\omega_d\chi h_v$,
    \begin{equation}
        \mathcal{P}_{r,k} = \frac{d(1-\omega_d)^d\varepsilon_k^{-1}}{N_k^d}b_v^dZ_{G\setminus\mathcal{N}[v]}\left(-\varepsilon_k^{-d}\right)\left[\frac{\partial Z_G\left(-\varepsilon^{-d}\right)}{\partial\varepsilon}\bigg|_{\varepsilon=\varepsilon_k}\right]^{d-2}T_G^{\succ r}\left(\varepsilon_k^{-1}\right)\frac{\partial T_G^r\left(\varepsilon^{-1}\right)}{\partial\varepsilon}\bigg|_{\varepsilon=\varepsilon_k}. \notag
    \end{equation}
    Now, by applying Proposition~\ref{proposition:VertexRecursionRelation} and using $Z_G\left(-\varepsilon_k^{-d}\right)=0$,
    \begin{equation}
        \mathcal{P}_{r,k} = \frac{d(1-\omega_d)^d\varepsilon_k^{d-1}}{N_k^d}Z_{G\setminus\{v\}}\left(-\varepsilon_k^{-d}\right)\left[\frac{\partial Z_G\left(-\varepsilon^{-d}\right)}{\partial\varepsilon}\bigg|_{\varepsilon=\varepsilon_k}\right]^{d-2}T_G^{\succ r}\left(\varepsilon_k^{-1}\right)\frac{\partial T_G^r\left(\varepsilon^{-1}\right)}{\partial\varepsilon}\bigg|_{\varepsilon=\varepsilon_k}. \notag
    \end{equation}
    By setting
    \begin{equation}
        N_k = (1-\omega_d)\left[dZ_{G\setminus\{v\}}\left(-\varepsilon_k^{-d}\right)\left(\varepsilon_k\frac{\partial Z_G\left(-\varepsilon^{-d}\right)}{\partial\varepsilon}\bigg|_{\varepsilon=\varepsilon_k}\right)^{d-1}\right]^{\frac{1}{d}}, \notag
    \end{equation}
    we obtain
    \begin{equation}
        \mathcal{P}_{r,k} = \left[\frac{\partial Z_G\left(-\varepsilon^{-d}\right)}{\partial\varepsilon}\bigg|_{\varepsilon=\varepsilon_k}\right]^{-1}T_G^{\succ r}\left(\varepsilon_k^{-1}\right)\frac{\partial T_G^r\left(\varepsilon^{-1}\right)}{\partial\varepsilon}\bigg|_{\varepsilon=\varepsilon_k}. \notag
    \end{equation}
    By applying Lemma~\ref{lemma:TransferOperatorDerivativeIdempotenceRelation},
    \begin{align}
        \mathcal{P}_{r,k}^2 &= \left[\frac{\partial Z_G\left(-\varepsilon^{-d}\right)}{\partial\varepsilon}\bigg|_{\varepsilon=\varepsilon_k}\right]^{-1}T_G^{\succ r}\left(\varepsilon_k^{-1}\right)\frac{\partial T_G^r\left(\varepsilon^{-1}\right)}{\partial\varepsilon}\bigg|_{\varepsilon=\varepsilon_k} \notag \\
        &= \mathcal{P}_{r,k}. \notag
    \end{align}
    By applying Lemma~\ref{lemma:TransferOperatorParafermionicModeCommutator} and using Corollary~\ref{corollary:HamiltonianTransferOperatorCommutator}, we obtain
    \begin{align}
        \mathcal{P}_{r,k}\mathcal{P}_{s,k} &= \left(\prod_{p=1}^d\psi_{r-p,k}\right)\left(\prod_{p=1}^d\psi_{s-p,k}\right) \notag \\
        &= \frac{1}{N_k}\left(\prod_{p=1}^{d-1}\psi_{r-p,k}\right)T_G^r\left(\varepsilon_k^{-1}\right) \chi T_G^{\succ r}\left(\varepsilon_k^{-1}\right)\left(\prod_{p=1}^d\psi_{s-p,k}\right) \notag \\
        &= 0, \notag
    \end{align}
    for $r$ and $s$ distinct. This completes the proof.
\end{proof}

\section{Proof of Lemma~\ref*{lemma:HamiltonianProjectorOperatorRelation}}
\label{section:HamiltonianProjectorOperatorRelation}

\HamiltonianProjectorOperatorRelation*

\begin{proof}
    By the factor theorem, for all $k\in[\alpha(G)]$,
    \begin{equation}
        \frac{\partial Z_G\left(-\varepsilon^{-d}\right)}{\partial\varepsilon}\bigg|_{\varepsilon=\varepsilon_k} = d\varepsilon_k^{-1}\prod_{\substack{j=1 \\ j \neq k}}^{\alpha(G)}\left(1-\varepsilon_j^d\varepsilon_k^{-d}\right). \notag
    \end{equation}
    By applying Lemma~\ref{lemma:ProjectorOperatorRelations},
    \begin{align}
        \sum_{k=1}^{\alpha(G)}\sum_{r\in\mathbb{Z}_d}\omega_d^r\varepsilon_k\mathcal{P}_{r,k} &= \frac{1}{d}\sum_{k=1}^{\alpha(G)}\sum_{r\in\mathbb{Z}_d}\omega_d^r\varepsilon_k^2\prod_{\substack{j=1 \\ j \neq k}}^{\alpha(G)}\left(\frac{\varepsilon_k^d}{\varepsilon_k^d-\varepsilon_j^d}\right)T_G^{\succ r}\left(\varepsilon_k^{-1}\right)\frac{\partial T_G^r\left(\varepsilon^{-1}\right)}{\partial\varepsilon}\bigg|_{\varepsilon=\varepsilon_k} \notag \\
        &= \frac{1}{d}\sum_{k=1}^{\alpha(G)}\sum_{r\in\mathbb{Z}_d}\omega_d^r\varepsilon_k^2\prod_{\substack{j=1 \\ j \neq k}}^{\alpha(G)}\left(\frac{\varepsilon_k^d}{\varepsilon_k^d-\varepsilon_j^d}\right)T_G^{\succ 0}\left(\omega_d^{-r}\varepsilon_k^{-1}\right)\frac{\partial T_G^0\left(\omega_d^{-r}\varepsilon^{-1}\right)}{\partial\varepsilon}\bigg|_{\varepsilon=\varepsilon_k} \notag \\
        &= -\frac{1}{d}\sum_{k=1}^{\alpha(G)}\sum_{r\in\mathbb{Z}_d}\prod_{\substack{j=1 \\ j \neq k}}^{\alpha(G)}\left(\frac{\varepsilon_k^d}{\varepsilon_k^d-\varepsilon_j^d}\right)\left(T_G^{\succ 0}(u)\frac{\partial T_G^0(u)}{\partial u}\right)\bigg|_{u=\omega_d^{-r}\varepsilon_k^{-1}}. \notag
    \end{align}
    Finally, by the Lagrange interpolation formula,
    \begin{equation}
        \sum_{k=1}^{\alpha(G)}\sum_{r\in\mathbb{Z}_d}\omega_d^r\varepsilon_k\mathcal{P}_{r,k} = -\left(T_G^{\succ 0}(u)\frac{\partial T_G^0(u)}{\partial u}\right)\bigg|_{u=0} = H, \notag
    \end{equation}
    completing the proof.
\end{proof}

\bibliography{bibliography}

\begin{thebibliography}{46}%
\makeatletter
\providecommand \@ifxundefined [1]{%
 \@ifx{#1\undefined}
}%
\providecommand \@ifnum [1]{%
 \ifnum #1\expandafter \@firstoftwo
 \else \expandafter \@secondoftwo
 \fi
}%
\providecommand \@ifx [1]{%
 \ifx #1\expandafter \@firstoftwo
 \else \expandafter \@secondoftwo
 \fi
}%
\providecommand \natexlab [1]{#1}%
\providecommand \enquote  [1]{``#1''}%
\providecommand \bibnamefont  [1]{#1}%
\providecommand \bibfnamefont [1]{#1}%
\providecommand \citenamefont [1]{#1}%
\providecommand \href@noop [0]{\@secondoftwo}%
\providecommand \href [0]{\begingroup \@sanitize@url \@href}%
\providecommand \@href[1]{\@@startlink{#1}\@@href}%
\providecommand \@@href[1]{\endgroup#1\@@endlink}%
\providecommand \@sanitize@url [0]{\catcode `\\12\catcode `\$12\catcode `\&12\catcode `\#12\catcode `\^12\catcode `\_12\catcode `\%12\relax}%
\providecommand \@@startlink[1]{}%
\providecommand \@@endlink[0]{}%
\providecommand \url  [0]{\begingroup\@sanitize@url \@url }%
\providecommand \@url [1]{\endgroup\@href {#1}{\urlprefix }}%
\providecommand \urlprefix  [0]{URL }%
\providecommand \Eprint [0]{\href }%
\providecommand \doibase [0]{https://doi.org/}%
\providecommand \selectlanguage [0]{\@gobble}%
\providecommand \bibinfo  [0]{\@secondoftwo}%
\providecommand \bibfield  [0]{\@secondoftwo}%
\providecommand \translation [1]{[#1]}%
\providecommand \BibitemOpen [0]{}%
\providecommand \bibitemStop [0]{}%
\providecommand \bibitemNoStop [0]{.\EOS\space}%
\providecommand \EOS [0]{\spacefactor3000\relax}%
\providecommand \BibitemShut  [1]{\csname bibitem#1\endcsname}%
\let\auto@bib@innerbib\@empty
\bibitem [{\citenamefont {Onsager}(1944)}]{onsager1944crystal}%
  \BibitemOpen
  \bibfield  {author} {\bibinfo {author} {\bibfnamefont {L.}~\bibnamefont {Onsager}},\ }\href {https://doi.org/10.1103/PhysRev.65.117} {\bibfield  {journal} {\bibinfo  {journal} {Physical Review}\ }\textbf {\bibinfo {volume} {65}},\ \bibinfo {pages} {117} (\bibinfo {year} {1944})}\BibitemShut {NoStop}%
\bibitem [{\citenamefont {Kaufman}(1949)}]{kaufman1949crystal}%
  \BibitemOpen
  \bibfield  {author} {\bibinfo {author} {\bibfnamefont {B.}~\bibnamefont {Kaufman}},\ }\href {https://doi.org/10.1103/PhysRev.76.1232} {\bibfield  {journal} {\bibinfo  {journal} {Physical Review}\ }\textbf {\bibinfo {volume} {76}},\ \bibinfo {pages} {1232} (\bibinfo {year} {1949})}\BibitemShut {NoStop}%
\bibitem [{\citenamefont {Jordan}\ and\ \citenamefont {Wigner}(1928)}]{jordan1928uber}%
  \BibitemOpen
  \bibfield  {author} {\bibinfo {author} {\bibfnamefont {P.}~\bibnamefont {Jordan}}\ and\ \bibinfo {author} {\bibfnamefont {E.}~\bibnamefont {Wigner}},\ }\href {https://doi.org/10.1007/bf01331938} {\bibfield  {journal} {\bibinfo  {journal} {Z. Phys}\ }\textbf {\bibinfo {volume} {47}},\ \bibinfo {pages} {631} (\bibinfo {year} {1928})}\BibitemShut {NoStop}%
\bibitem [{\citenamefont {Schultz}\ \emph {et~al.}(1964)\citenamefont {Schultz}, \citenamefont {Mattis},\ and\ \citenamefont {Lieb}}]{schultz1964two}%
  \BibitemOpen
  \bibfield  {author} {\bibinfo {author} {\bibfnamefont {T.~D.}\ \bibnamefont {Schultz}}, \bibinfo {author} {\bibfnamefont {D.~C.}\ \bibnamefont {Mattis}},\ and\ \bibinfo {author} {\bibfnamefont {E.~H.}\ \bibnamefont {Lieb}},\ }\href {https://doi.org/10.1103/RevModPhys.36.856} {\bibfield  {journal} {\bibinfo  {journal} {Reviews of Modern Physics}\ }\textbf {\bibinfo {volume} {36}},\ \bibinfo {pages} {856} (\bibinfo {year} {1964})}\BibitemShut {NoStop}%
\bibitem [{\citenamefont {Kitaev}(2006)}]{kitaev2006anyons}%
  \BibitemOpen
  \bibfield  {author} {\bibinfo {author} {\bibfnamefont {A.}~\bibnamefont {Kitaev}},\ }\href {https://doi.org/10.1016/j.aop.2005.10.005} {\bibfield  {journal} {\bibinfo  {journal} {Annals of Physics}\ }\textbf {\bibinfo {volume} {321}},\ \bibinfo {pages} {2} (\bibinfo {year} {2006})},\ \Eprint {https://arxiv.org/abs/cond-mat/0506438} {arXiv:cond-mat/0506438} \BibitemShut {NoStop}%
\bibitem [{\citenamefont {Mandal}\ and\ \citenamefont {Surendran}(2009)}]{mandal2009exactly}%
  \BibitemOpen
  \bibfield  {author} {\bibinfo {author} {\bibfnamefont {S.}~\bibnamefont {Mandal}}\ and\ \bibinfo {author} {\bibfnamefont {N.}~\bibnamefont {Surendran}},\ }\href {https://doi.org/10.1103/PhysRevB.79.024426} {\bibfield  {journal} {\bibinfo  {journal} {Physical Review B}\ }\textbf {\bibinfo {volume} {79}},\ \bibinfo {pages} {024426} (\bibinfo {year} {2009})},\ \Eprint {https://arxiv.org/abs/0801.0229} {arXiv:0801.0229} \BibitemShut {NoStop}%
\bibitem [{\citenamefont {Chapman}\ and\ \citenamefont {Flammia}(2020)}]{chapman2020characterization}%
  \BibitemOpen
  \bibfield  {author} {\bibinfo {author} {\bibfnamefont {A.}~\bibnamefont {Chapman}}\ and\ \bibinfo {author} {\bibfnamefont {S.~T.}\ \bibnamefont {Flammia}},\ }\href {https://doi.org/10.22331/q-2020-06-04-278} {\bibfield  {journal} {\bibinfo  {journal} {Quantum}\ }\textbf {\bibinfo {volume} {4}},\ \bibinfo {pages} {278} (\bibinfo {year} {2020})},\ \Eprint {https://arxiv.org/abs/2003.05465} {arXiv:2003.05465} \BibitemShut {NoStop}%
\bibitem [{\citenamefont {Ogura}\ \emph {et~al.}(2020)\citenamefont {Ogura}, \citenamefont {Imamura}, \citenamefont {Kameyama}, \citenamefont {Minami},\ and\ \citenamefont {Sato}}]{ogura2020geometric}%
  \BibitemOpen
  \bibfield  {author} {\bibinfo {author} {\bibfnamefont {M.}~\bibnamefont {Ogura}}, \bibinfo {author} {\bibfnamefont {Y.}~\bibnamefont {Imamura}}, \bibinfo {author} {\bibfnamefont {N.}~\bibnamefont {Kameyama}}, \bibinfo {author} {\bibfnamefont {K.}~\bibnamefont {Minami}},\ and\ \bibinfo {author} {\bibfnamefont {M.}~\bibnamefont {Sato}},\ }\href {https://doi.org/10.1103/PhysRevB.102.245118} {\bibfield  {journal} {\bibinfo  {journal} {Physical Review B}\ }\textbf {\bibinfo {volume} {102}},\ \bibinfo {pages} {245118} (\bibinfo {year} {2020})}\BibitemShut {NoStop}%
\bibitem [{\citenamefont {Fendley}(2019)}]{fendley2019free}%
  \BibitemOpen
  \bibfield  {author} {\bibinfo {author} {\bibfnamefont {P.}~\bibnamefont {Fendley}},\ }\href {https://doi.org/10.1088/1751-8121/ab305d} {\bibfield  {journal} {\bibinfo  {journal} {Journal of Physics A: Mathematical and Theoretical}\ }\textbf {\bibinfo {volume} {52}},\ \bibinfo {pages} {335002} (\bibinfo {year} {2019})},\ \Eprint {https://arxiv.org/abs/1901.08078} {arXiv:1901.08078} \BibitemShut {NoStop}%
\bibitem [{\citenamefont {Elman}\ \emph {et~al.}(2021)\citenamefont {Elman}, \citenamefont {Chapman},\ and\ \citenamefont {Flammia}}]{elman2021free}%
  \BibitemOpen
  \bibfield  {author} {\bibinfo {author} {\bibfnamefont {S.~J.}\ \bibnamefont {Elman}}, \bibinfo {author} {\bibfnamefont {A.}~\bibnamefont {Chapman}},\ and\ \bibinfo {author} {\bibfnamefont {S.~T.}\ \bibnamefont {Flammia}},\ }\href {https://doi.org/10.1007/s00220-021-04220-w} {\bibfield  {journal} {\bibinfo  {journal} {Communications in Mathematical Physics}\ }\textbf {\bibinfo {volume} {388}},\ \bibinfo {pages} {969} (\bibinfo {year} {2021})},\ \Eprint {https://arxiv.org/abs/2012.07857} {arXiv:2012.07857} \BibitemShut {NoStop}%
\bibitem [{\citenamefont {Chapman}\ \emph {et~al.}(2023)\citenamefont {Chapman}, \citenamefont {Elman},\ and\ \citenamefont {Mann}}]{chapman2023unified}%
  \BibitemOpen
  \bibfield  {author} {\bibinfo {author} {\bibfnamefont {A.}~\bibnamefont {Chapman}}, \bibinfo {author} {\bibfnamefont {S.~J.}\ \bibnamefont {Elman}},\ and\ \bibinfo {author} {\bibfnamefont {R.~L.}\ \bibnamefont {Mann}},\ }\href@noop {} {\bibfield  {journal} {\bibinfo  {journal} {arXiv e-prints}\ } (\bibinfo {year} {2023})},\ \Eprint {https://arxiv.org/abs/2305.15625} {arXiv:2305.15625} \BibitemShut {NoStop}%
\bibitem [{\citenamefont {Fendley}\ and\ \citenamefont {Pozsgay}(2024)}]{fendley2024free}%
  \BibitemOpen
  \bibfield  {author} {\bibinfo {author} {\bibfnamefont {P.}~\bibnamefont {Fendley}}\ and\ \bibinfo {author} {\bibfnamefont {B.}~\bibnamefont {Pozsgay}},\ }\href {https://doi.org/10.21468/SciPostPhys.16.4.102} {\bibfield  {journal} {\bibinfo  {journal} {SciPost Physics}\ }\textbf {\bibinfo {volume} {16}},\ \bibinfo {pages} {102} (\bibinfo {year} {2024})},\ \Eprint {https://arxiv.org/abs/2310.19897} {arXiv:2310.19897} \BibitemShut {NoStop}%
\bibitem [{\citenamefont {Alicea}\ and\ \citenamefont {Fendley}(2016)}]{alicea2016topological}%
  \BibitemOpen
  \bibfield  {author} {\bibinfo {author} {\bibfnamefont {J.}~\bibnamefont {Alicea}}\ and\ \bibinfo {author} {\bibfnamefont {P.}~\bibnamefont {Fendley}},\ }\href {https://doi.org/10.1146/annurev-conmatphys-031115-011336} {\bibfield  {journal} {\bibinfo  {journal} {Annual Review of Condensed Matter Physics}\ }\textbf {\bibinfo {volume} {7}},\ \bibinfo {pages} {119} (\bibinfo {year} {2016})},\ \Eprint {https://arxiv.org/abs/1504.02476} {arXiv:1504.02476} \BibitemShut {NoStop}%
\bibitem [{\citenamefont {Fradkin}\ and\ \citenamefont {Kadanoff}(1980)}]{fradkin1980disorder}%
  \BibitemOpen
  \bibfield  {author} {\bibinfo {author} {\bibfnamefont {E.}~\bibnamefont {Fradkin}}\ and\ \bibinfo {author} {\bibfnamefont {L.~P.}\ \bibnamefont {Kadanoff}},\ }\href {https://doi.org/10.1016/0550-3213(80)90472-1} {\bibfield  {journal} {\bibinfo  {journal} {Nuclear Physics B}\ }\textbf {\bibinfo {volume} {170}},\ \bibinfo {pages} {1} (\bibinfo {year} {1980})}\BibitemShut {NoStop}%
\bibitem [{\citenamefont {Fendley}(2012)}]{fendley2012parafermionic}%
  \BibitemOpen
  \bibfield  {author} {\bibinfo {author} {\bibfnamefont {P.}~\bibnamefont {Fendley}},\ }\href {https://doi.org/10.1088/1742-5468/2012/11/P11020} {\bibfield  {journal} {\bibinfo  {journal} {Journal of Statistical Mechanics: Theory and Experiment}\ }\textbf {\bibinfo {volume} {2012}},\ \bibinfo {pages} {P11020} (\bibinfo {year} {2012})},\ \Eprint {https://arxiv.org/abs/1209.0472} {arXiv:1209.0472} \BibitemShut {NoStop}%
\bibitem [{\citenamefont {von Gehlen}\ and\ \citenamefont {Rittenberg}(1985)}]{vongehlen1985zn}%
  \BibitemOpen
  \bibfield  {author} {\bibinfo {author} {\bibfnamefont {G.}~\bibnamefont {von Gehlen}}\ and\ \bibinfo {author} {\bibfnamefont {V.}~\bibnamefont {Rittenberg}},\ }\href {https://doi.org/10.1016/0550-3213(85)90350-5} {\bibfield  {journal} {\bibinfo  {journal} {Nuclear Physics B}\ }\textbf {\bibinfo {volume} {257}},\ \bibinfo {pages} {351} (\bibinfo {year} {1985})}\BibitemShut {NoStop}%
\bibitem [{\citenamefont {Baxter}(1988)}]{baxter1988free}%
  \BibitemOpen
  \bibfield  {author} {\bibinfo {author} {\bibfnamefont {R.}~\bibnamefont {Baxter}},\ }\href {https://doi.org/10.1007/BF01019722} {\bibfield  {journal} {\bibinfo  {journal} {Journal of Statistical Physics}\ }\textbf {\bibinfo {volume} {52}},\ \bibinfo {pages} {639} (\bibinfo {year} {1988})}\BibitemShut {NoStop}%
\bibitem [{\citenamefont {Au-Yang}\ and\ \citenamefont {Perk}(1997)}]{auyang1997many}%
  \BibitemOpen
  \bibfield  {author} {\bibinfo {author} {\bibfnamefont {H.}~\bibnamefont {Au-Yang}}\ and\ \bibinfo {author} {\bibfnamefont {J.~H.}\ \bibnamefont {Perk}},\ }\href {https://doi.org/10.1142/S0217979297000046} {\bibfield  {journal} {\bibinfo  {journal} {International Journal of Modern Physics B}\ }\textbf {\bibinfo {volume} {11}},\ \bibinfo {pages} {11} (\bibinfo {year} {1997})},\ \Eprint {https://arxiv.org/abs/q-alg/9609003} {arXiv:q-alg/9609003} \BibitemShut {NoStop}%
\bibitem [{\citenamefont {Baxter}(2006)}]{baxter2006challenge}%
  \BibitemOpen
  \bibfield  {author} {\bibinfo {author} {\bibfnamefont {R.}~\bibnamefont {Baxter}},\ }in\ \href {https://doi.org/10.1088/1742-6596/42/1/003} {\emph {\bibinfo {booktitle} {Journal of Physics: Conference Series}}},\ Vol.~\bibinfo {volume} {42}\ (\bibinfo {organization} {IOP Publishing},\ \bibinfo {year} {2006})\ p.~\bibinfo {pages} {11},\ \Eprint {https://arxiv.org/abs/cond-mat/0510683} {arXiv:cond-mat/0510683} \BibitemShut {NoStop}%
\bibitem [{\citenamefont {Baxter}(1989)}]{baxter1989simple}%
  \BibitemOpen
  \bibfield  {author} {\bibinfo {author} {\bibfnamefont {R.}~\bibnamefont {Baxter}},\ }\href {https://doi.org/10.1016/0375-9601(89)90884-0} {\bibfield  {journal} {\bibinfo  {journal} {Physics Letters A}\ }\textbf {\bibinfo {volume} {140}},\ \bibinfo {pages} {155} (\bibinfo {year} {1989})}\BibitemShut {NoStop}%
\bibitem [{\citenamefont {Bazhanov}\ and\ \citenamefont {Stroganov}(1990)}]{bazhanov1990chiral}%
  \BibitemOpen
  \bibfield  {author} {\bibinfo {author} {\bibfnamefont {V.}~\bibnamefont {Bazhanov}}\ and\ \bibinfo {author} {\bibfnamefont {Y.~G.}\ \bibnamefont {Stroganov}},\ }\href {https://doi.org/10.1007/BF01025851} {\bibfield  {journal} {\bibinfo  {journal} {Journal of Statistical Physics}\ }\textbf {\bibinfo {volume} {59}},\ \bibinfo {pages} {799} (\bibinfo {year} {1990})}\BibitemShut {NoStop}%
\bibitem [{\citenamefont {Baxter}(2004)}]{baxter2004transfer}%
  \BibitemOpen
  \bibfield  {author} {\bibinfo {author} {\bibfnamefont {R.}~\bibnamefont {Baxter}},\ }\href {https://doi.org/10.1023/B:JOSS.0000044062.64287.b9} {\bibfield  {journal} {\bibinfo  {journal} {Journal of Statistical Physics}\ }\textbf {\bibinfo {volume} {117}},\ \bibinfo {pages} {1} (\bibinfo {year} {2004})},\ \Eprint {https://arxiv.org/abs/cond-mat/0409493} {arXiv:cond-mat/0409493} \BibitemShut {NoStop}%
\bibitem [{\citenamefont {Fendley}(2014)}]{fendley2014free}%
  \BibitemOpen
  \bibfield  {author} {\bibinfo {author} {\bibfnamefont {P.}~\bibnamefont {Fendley}},\ }\href {https://doi.org/10.1088/1751-8113/47/7/075001} {\bibfield  {journal} {\bibinfo  {journal} {Journal of Physics A: Mathematical and Theoretical}\ }\textbf {\bibinfo {volume} {47}},\ \bibinfo {pages} {075001} (\bibinfo {year} {2014})},\ \Eprint {https://arxiv.org/abs/1310.6049} {arXiv:1310.6049} \BibitemShut {NoStop}%
\bibitem [{\citenamefont {Baxter}(2014)}]{baxter2014tau2}%
  \BibitemOpen
  \bibfield  {author} {\bibinfo {author} {\bibfnamefont {R.}~\bibnamefont {Baxter}},\ }\href {https://doi.org/10.1088/1751-8113/47/31/315001} {\bibfield  {journal} {\bibinfo  {journal} {Journal of Physics A: Mathematical and Theoretical}\ }\textbf {\bibinfo {volume} {47}},\ \bibinfo {pages} {315001} (\bibinfo {year} {2014})},\ \Eprint {https://arxiv.org/abs/1310.7074} {arXiv:1310.7074} \BibitemShut {NoStop}%
\bibitem [{\citenamefont {Au-Yang}\ and\ \citenamefont {Perk}(2014)}]{au2014parafermions}%
  \BibitemOpen
  \bibfield  {author} {\bibinfo {author} {\bibfnamefont {H.}~\bibnamefont {Au-Yang}}\ and\ \bibinfo {author} {\bibfnamefont {J.~H.}\ \bibnamefont {Perk}},\ }\href {https://doi.org/10.1088/1751-8113/47/31/315002} {\bibfield  {journal} {\bibinfo  {journal} {Journal of Physics A: Mathematical and Theoretical}\ }\textbf {\bibinfo {volume} {47}},\ \bibinfo {pages} {315002} (\bibinfo {year} {2014})},\ \Eprint {https://arxiv.org/abs/1402.0061} {arXiv:1402.0061} \BibitemShut {NoStop}%
\bibitem [{\citenamefont {Au-Yang}\ and\ \citenamefont {Perk}(2016)}]{au2016parafermions}%
  \BibitemOpen
  \bibfield  {author} {\bibinfo {author} {\bibfnamefont {H.}~\bibnamefont {Au-Yang}}\ and\ \bibinfo {author} {\bibfnamefont {J.~H.}\ \bibnamefont {Perk}},\ }\href@noop {} {\bibfield  {journal} {\bibinfo  {journal} {arXiv e-prints}\ } (\bibinfo {year} {2016})},\ \Eprint {https://arxiv.org/abs/1606.06319} {arXiv:1606.06319} \BibitemShut {NoStop}%
\bibitem [{\citenamefont {Yamazaki}(1964)}]{yamazaki1964projective}%
  \BibitemOpen
  \bibfield  {author} {\bibinfo {author} {\bibfnamefont {K.}~\bibnamefont {Yamazaki}},\ }\href@noop {} {\bibfield  {journal} {\bibinfo  {journal} {J. Fac. Sci. Univ. Tokyo Sect. I}\ }\textbf {\bibinfo {volume} {10}},\ \bibinfo {pages} {1964} (\bibinfo {year} {1964})}\BibitemShut {NoStop}%
\bibitem [{\citenamefont {Morris}(1967)}]{morris1967generalized}%
  \BibitemOpen
  \bibfield  {author} {\bibinfo {author} {\bibfnamefont {A.}~\bibnamefont {Morris}},\ }\href {https://doi.org/10.1093/qmath/18.1.7} {\bibfield  {journal} {\bibinfo  {journal} {The Quarterly Journal of Mathematics}\ }\textbf {\bibinfo {volume} {18}},\ \bibinfo {pages} {7} (\bibinfo {year} {1967})}\BibitemShut {NoStop}%
\bibitem [{\citenamefont {Alcaraz}\ \emph {et~al.}(2017)\citenamefont {Alcaraz}, \citenamefont {Batchelor},\ and\ \citenamefont {Liu}}]{alcaraz2017energy}%
  \BibitemOpen
  \bibfield  {author} {\bibinfo {author} {\bibfnamefont {F.~C.}\ \bibnamefont {Alcaraz}}, \bibinfo {author} {\bibfnamefont {M.~T.}\ \bibnamefont {Batchelor}},\ and\ \bibinfo {author} {\bibfnamefont {Z.-Z.}\ \bibnamefont {Liu}},\ }\href {https://doi.org/10.1088/1751-8121/aa645a} {\bibfield  {journal} {\bibinfo  {journal} {Journal of Physics A: Mathematical and Theoretical}\ }\textbf {\bibinfo {volume} {50}},\ \bibinfo {pages} {16LT03} (\bibinfo {year} {2017})},\ \Eprint {https://arxiv.org/abs/1612.02617} {arXiv:1612.02617} \BibitemShut {NoStop}%
\bibitem [{\citenamefont {Alcaraz}\ and\ \citenamefont {Batchelor}(2018)}]{alcaraz2018anomalous}%
  \BibitemOpen
  \bibfield  {author} {\bibinfo {author} {\bibfnamefont {F.~C.}\ \bibnamefont {Alcaraz}}\ and\ \bibinfo {author} {\bibfnamefont {M.~T.}\ \bibnamefont {Batchelor}},\ }\href {https://doi.org/10.1103/PhysRevE.97.062118} {\bibfield  {journal} {\bibinfo  {journal} {Physical Review E}\ }\textbf {\bibinfo {volume} {97}},\ \bibinfo {pages} {062118} (\bibinfo {year} {2018})},\ \Eprint {https://arxiv.org/abs/1802.04453} {arXiv:1802.04453} \BibitemShut {NoStop}%
\bibitem [{\citenamefont {Alcaraz}\ and\ \citenamefont {Pimenta}(2020{\natexlab{a}})}]{alcaraz2020integrable}%
  \BibitemOpen
  \bibfield  {author} {\bibinfo {author} {\bibfnamefont {F.~C.}\ \bibnamefont {Alcaraz}}\ and\ \bibinfo {author} {\bibfnamefont {R.~A.}\ \bibnamefont {Pimenta}},\ }\href {https://doi.org/10.1103/PhysRevB.102.235170} {\bibfield  {journal} {\bibinfo  {journal} {Physical Review B}\ }\textbf {\bibinfo {volume} {102}},\ \bibinfo {pages} {235170} (\bibinfo {year} {2020}{\natexlab{a}})},\ \Eprint {https://arxiv.org/abs/2010.01116} {arXiv:2010.01116} \BibitemShut {NoStop}%
\bibitem [{\citenamefont {Alcaraz}\ and\ \citenamefont {Pimenta}(2020{\natexlab{b}})}]{alcaraz2020free}%
  \BibitemOpen
  \bibfield  {author} {\bibinfo {author} {\bibfnamefont {F.~C.}\ \bibnamefont {Alcaraz}}\ and\ \bibinfo {author} {\bibfnamefont {R.~A.}\ \bibnamefont {Pimenta}},\ }\href {https://doi.org/10.1103/physrevb.102.121101} {\bibfield  {journal} {\bibinfo  {journal} {Physical Review B}\ }\textbf {\bibinfo {volume} {102}},\ \bibinfo {pages} {121101} (\bibinfo {year} {2020}{\natexlab{b}})},\ \Eprint {https://arxiv.org/abs/2005.14622} {arXiv:2005.14622} \BibitemShut {NoStop}%
\bibitem [{\citenamefont {Alcaraz}\ and\ \citenamefont {Pimenta}(2021)}]{alcaraz2021free}%
  \BibitemOpen
  \bibfield  {author} {\bibinfo {author} {\bibfnamefont {F.~C.}\ \bibnamefont {Alcaraz}}\ and\ \bibinfo {author} {\bibfnamefont {R.~A.}\ \bibnamefont {Pimenta}},\ }\href {https://doi.org/10.1103/PhysRevE.104.054121} {\bibfield  {journal} {\bibinfo  {journal} {Physical Review E}\ }\textbf {\bibinfo {volume} {104}},\ \bibinfo {pages} {054121} (\bibinfo {year} {2021})},\ \Eprint {https://arxiv.org/abs/2108.04372} {arXiv:2108.04372} \BibitemShut {NoStop}%
\bibitem [{\citenamefont {Alcaraz}\ \emph {et~al.}(2021)\citenamefont {Alcaraz}, \citenamefont {Hoyos},\ and\ \citenamefont {Pimenta}}]{alcaraz2021powerful}%
  \BibitemOpen
  \bibfield  {author} {\bibinfo {author} {\bibfnamefont {F.~C.}\ \bibnamefont {Alcaraz}}, \bibinfo {author} {\bibfnamefont {J.~A.}\ \bibnamefont {Hoyos}},\ and\ \bibinfo {author} {\bibfnamefont {R.~A.}\ \bibnamefont {Pimenta}},\ }\href {https://doi.org/10.1103/PhysRevB.104.174206} {\bibfield  {journal} {\bibinfo  {journal} {Physical Review B}\ }\textbf {\bibinfo {volume} {104}},\ \bibinfo {pages} {174206} (\bibinfo {year} {2021})},\ \Eprint {https://arxiv.org/abs/2109.01938} {arXiv:2109.01938} \BibitemShut {NoStop}%
\bibitem [{\citenamefont {Alcaraz}\ \emph {et~al.}(2023)\citenamefont {Alcaraz}, \citenamefont {Pimenta},\ and\ \citenamefont {Sirker}}]{alcaraz2023ising}%
  \BibitemOpen
  \bibfield  {author} {\bibinfo {author} {\bibfnamefont {F.~C.}\ \bibnamefont {Alcaraz}}, \bibinfo {author} {\bibfnamefont {R.~A.}\ \bibnamefont {Pimenta}},\ and\ \bibinfo {author} {\bibfnamefont {J.}~\bibnamefont {Sirker}},\ }\href {https://doi.org/10.1103/PhysRevB.107.235136} {\bibfield  {journal} {\bibinfo  {journal} {Physical Review B}\ }\textbf {\bibinfo {volume} {107}},\ \bibinfo {pages} {235136} (\bibinfo {year} {2023})},\ \Eprint {https://arxiv.org/abs/2303.15284} {arXiv:2303.15284} \BibitemShut {NoStop}%
\bibitem [{\citenamefont {Batchelor}\ \emph {et~al.}(2023)\citenamefont {Batchelor}, \citenamefont {Henry},\ and\ \citenamefont {Lu}}]{batchelor2023brief}%
  \BibitemOpen
  \bibfield  {author} {\bibinfo {author} {\bibfnamefont {M.~T.}\ \bibnamefont {Batchelor}}, \bibinfo {author} {\bibfnamefont {R.~A.}\ \bibnamefont {Henry}},\ and\ \bibinfo {author} {\bibfnamefont {X.}~\bibnamefont {Lu}},\ }\href {https://doi.org/10.1007/s43673-023-00105-3} {\bibfield  {journal} {\bibinfo  {journal} {AAPPS Bulletin}\ }\textbf {\bibinfo {volume} {33}},\ \bibinfo {pages} {29} (\bibinfo {year} {2023})},\ \Eprint {https://arxiv.org/abs/2311.04582} {arXiv:2311.04582} \BibitemShut {NoStop}%
\bibitem [{\citenamefont {Rose}(1970)}]{rose1970triangulated}%
  \BibitemOpen
  \bibfield  {author} {\bibinfo {author} {\bibfnamefont {D.~J.}\ \bibnamefont {Rose}},\ }\href {https://doi.org/10.1016/0022-247X(70)90282-9} {\bibfield  {journal} {\bibinfo  {journal} {Journal of Mathematical Analysis and Applications}\ }\textbf {\bibinfo {volume} {32}},\ \bibinfo {pages} {597} (\bibinfo {year} {1970})}\BibitemShut {NoStop}%
\bibitem [{\citenamefont {Looges}\ and\ \citenamefont {Olariu}(1993)}]{looges1993optimal}%
  \BibitemOpen
  \bibfield  {author} {\bibinfo {author} {\bibfnamefont {P.~J.}\ \bibnamefont {Looges}}\ and\ \bibinfo {author} {\bibfnamefont {S.}~\bibnamefont {Olariu}},\ }\href {https://doi.org/10.1016/0898-1221(93)90308-I} {\bibfield  {journal} {\bibinfo  {journal} {Computers \& Mathematics with Applications}\ }\textbf {\bibinfo {volume} {25}},\ \bibinfo {pages} {15} (\bibinfo {year} {1993})}\BibitemShut {NoStop}%
\bibitem [{\citenamefont {Achlioptas}\ and\ \citenamefont {Zampetakis}(2021)}]{achlioptas2021local}%
  \BibitemOpen
  \bibfield  {author} {\bibinfo {author} {\bibfnamefont {D.}~\bibnamefont {Achlioptas}}\ and\ \bibinfo {author} {\bibfnamefont {K.}~\bibnamefont {Zampetakis}},\ }in\ \href {https://doi.org/10.4230/LIPIcs.ICALP.2021.8} {\emph {\bibinfo {booktitle} {48th International Colloquium on Automata, Languages, and Programming}}}\ (\bibinfo {year} {2021})\BibitemShut {NoStop}%
\bibitem [{\citenamefont {Stoudenmire}\ \emph {et~al.}(2015)\citenamefont {Stoudenmire}, \citenamefont {Clarke}, \citenamefont {Mong},\ and\ \citenamefont {Alicea}}]{stoudenmire2015assembling}%
  \BibitemOpen
  \bibfield  {author} {\bibinfo {author} {\bibfnamefont {E.}~\bibnamefont {Stoudenmire}}, \bibinfo {author} {\bibfnamefont {D.~J.}\ \bibnamefont {Clarke}}, \bibinfo {author} {\bibfnamefont {R.~S.}\ \bibnamefont {Mong}},\ and\ \bibinfo {author} {\bibfnamefont {J.}~\bibnamefont {Alicea}},\ }\href {https://doi.org/10.1103/PhysRevB.91.235112} {\bibfield  {journal} {\bibinfo  {journal} {Physical Review B}\ }\textbf {\bibinfo {volume} {91}},\ \bibinfo {pages} {235112} (\bibinfo {year} {2015})},\ \Eprint {https://arxiv.org/abs/1501.05305} {arXiv:1501.05305} \BibitemShut {NoStop}%
\bibitem [{\citenamefont {Brandst{\"a}dt}\ and\ \citenamefont {Dragan}(2003)}]{brandstadt2003linear}%
  \BibitemOpen
  \bibfield  {author} {\bibinfo {author} {\bibfnamefont {A.}~\bibnamefont {Brandst{\"a}dt}}\ and\ \bibinfo {author} {\bibfnamefont {F.~F.}\ \bibnamefont {Dragan}},\ }\href {https://doi.org/10.1016/S0166-218X(02)00571-1} {\bibfield  {journal} {\bibinfo  {journal} {Discrete Applied Mathematics}\ }\textbf {\bibinfo {volume} {129}},\ \bibinfo {pages} {285} (\bibinfo {year} {2003})}\BibitemShut {NoStop}%
\bibitem [{\citenamefont {Fendley}(2023)}]{fendley2023personal}%
  \BibitemOpen
  \bibfield  {author} {\bibinfo {author} {\bibfnamefont {P.}~\bibnamefont {Fendley}},\ }\href@noop {} {\bibinfo {title} {Personal communication}} (\bibinfo {year} {2023})\BibitemShut {NoStop}%
\bibitem [{\citenamefont {Knill}(2001)}]{knill2001fermionic}%
  \BibitemOpen
  \bibfield  {author} {\bibinfo {author} {\bibfnamefont {E.}~\bibnamefont {Knill}},\ }\href@noop {} {\bibfield  {journal} {\bibinfo  {journal} {arXiv e-prints}\ } (\bibinfo {year} {2001})},\ \Eprint {https://arxiv.org/abs/quant-ph/0108033} {arXiv:quant-ph/0108033} \BibitemShut {NoStop}%
\bibitem [{\citenamefont {Terhal}\ and\ \citenamefont {DiVincenzo}(2002)}]{terhal2002classical}%
  \BibitemOpen
  \bibfield  {author} {\bibinfo {author} {\bibfnamefont {B.~M.}\ \bibnamefont {Terhal}}\ and\ \bibinfo {author} {\bibfnamefont {D.~P.}\ \bibnamefont {DiVincenzo}},\ }\href {https://doi.org/10.1103/physreva.65.032325} {\bibfield  {journal} {\bibinfo  {journal} {Physical Review A}\ }\textbf {\bibinfo {volume} {65}},\ \bibinfo {pages} {032325} (\bibinfo {year} {2002})},\ \Eprint {https://arxiv.org/abs/quant-ph/0108010} {arXiv:quant-ph/0108010} \BibitemShut {NoStop}%
\bibitem [{\citenamefont {Pozsgay}\ and\ \citenamefont {Fukai}(2024)}]{pozsgay2024quantum}%
  \BibitemOpen
  \bibfield  {author} {\bibinfo {author} {\bibfnamefont {B.}~\bibnamefont {Pozsgay}}\ and\ \bibinfo {author} {\bibfnamefont {K.}~\bibnamefont {Fukai}},\ }\href@noop {} {\bibfield  {journal} {\bibinfo  {journal} {arXiv e-prints}\ } (\bibinfo {year} {2024})},\ \Eprint {https://arxiv.org/abs/2402.02984} {arXiv:2402.02984} \BibitemShut {NoStop}%
\bibitem [{\citenamefont {Vona}\ \emph {et~al.}(2025)\citenamefont {Vona}, \citenamefont {Mesty{\'a}n},\ and\ \citenamefont {Pozsgay}}]{vona2025exact}%
  \BibitemOpen
  \bibfield  {author} {\bibinfo {author} {\bibfnamefont {I.}~\bibnamefont {Vona}}, \bibinfo {author} {\bibfnamefont {M.}~\bibnamefont {Mesty{\'a}n}},\ and\ \bibinfo {author} {\bibfnamefont {B.}~\bibnamefont {Pozsgay}},\ }\href {https://doi.org/10.1103/PhysRevB.111.144306} {\bibfield  {journal} {\bibinfo  {journal} {Physical Review B}\ }\textbf {\bibinfo {volume} {111}},\ \bibinfo {pages} {144306} (\bibinfo {year} {2025})},\ \Eprint {https://arxiv.org/abs/2405.20832} {arXiv:2405.20832} \BibitemShut {NoStop}%
\end{thebibliography}%

\end{document}